\documentclass[12pt,a4paper]{article}

\usepackage{amsmath,amssymb,amsthm}
\usepackage[margin=2.5cm]{geometry}
\usepackage{hyperref}
\usepackage{booktabs,multirow,array}
\usepackage{xcolor}
\usepackage{graphicx}
\usepackage{caption,subcaption}
\usepackage{enumitem}
\usepackage{parskip}
\usepackage{setspace}
\usepackage{titlesec}
\usepackage{fancyhdr}
\usepackage{tcolorbox}
\tcbuselibrary{skins,breakable}
\usepackage{siunitx}
\usepackage{mathtools}
\usepackage{natbib}
\usepackage{microtype}
\usepackage{bm}

\definecolor{trblue}{RGB}{26,86,219}
\definecolor{trred}{RGB}{224,36,36}
\definecolor{trgray}{RGB}{107,114,128}
\definecolor{trgreen}{RGB}{5,122,85}
\definecolor{trlb}{RGB}{219,234,254}
\definecolor{trlg}{RGB}{209,250,229}
\definecolor{trlr}{RGB}{254,226,226}

\titleformat{\section}{\large\bfseries\color{trblue}}{\thesection.}{0.55em}{}
\titleformat{\subsection}{\normalsize\bfseries}{\thesubsection.}{0.45em}{}
\titleformat{\subsubsection}{\normalsize\itshape}{\thesubsubsection.}{0.4em}{}
\titlespacing{\section}{0pt}{14pt}{5pt}
\titlespacing{\subsection}{0pt}{10pt}{3pt}

\newtcolorbox{resultboxM}{
  enhanced,breakable,colback=trlb,colframe=trblue,
  boxrule=1.2pt,arc=4pt,left=8pt,right=8pt,top=6pt,bottom=6pt,
  title={\small\bfseries\color{white}Manhattan (L\textsuperscript{1}) Model},
  colbacktitle=trblue,
}
\newtcolorbox{resultboxE}{
  enhanced,breakable,colback=trlr,colframe=trred,
  boxrule=1.2pt,arc=4pt,left=8pt,right=8pt,top=6pt,bottom=6pt,
  title={\small\bfseries\color{white}Euclidean (L\textsuperscript{2}) Model},
  colbacktitle=trred,
}
\newtcolorbox{intubox}[1]{
  enhanced,breakable,colback=trlg,colframe=trgreen,
  boxrule=0.9pt,arc=3pt,fontupper=\small,
  left=6pt,right=6pt,top=4pt,bottom=4pt,
  title={\small\bfseries #1},
  colbacktitle=trgreen,coltitle=white,
  attach boxed title to top left={yshift=-2mm,xshift=5mm},
}

\hypersetup{colorlinks=true,linkcolor=trblue!80!black,
            citecolor=trblue!80!black,urlcolor=trblue!80!black}

\theoremstyle{plain}
\newtheorem{theorem}{Theorem}
\newtheorem{proposition}{Proposition}
\theoremstyle{definition}
\newtheorem{definition}{Definition}
\theoremstyle{remark}
\newtheorem{remark}{Remark}

\newcommand{\W}{W_1}
\newcommand{\gM}{g_{\rm M}(\kappa)}
\newcommand{\EVKT}{\mathbb{E}[VKT]}
\newcommand{\dd}{\,\mathrm{d}}

\newtheorem{corollary}[theorem]{Corollary} 
\usepackage{float}
\usepackage{makecell}

\begin{document}

\begin{titlepage}
\vspace*{1.8cm}
{\centering
{\LARGE\bfseries
Approximation Models for Shared Mobility Rebalancing\\[6pt]
Under Structured Spatial Imbalance\\[6pt]
}


\hrule height 0.8pt
\vspace{7pt}
{\small\color{trgray}
}\vspace{7pt}
\hrule height 0.8pt}

\vspace{18pt}

{\centering
\textbf{Wenbo Fan}\textsuperscript{1,*} \quad \textbf{Zhouyun Chen}\textsuperscript{2}, \quad \textbf{Weihua Gu\textsuperscript{1}}

\vspace{8pt}
{\small
\textsuperscript{1}Department of Electrical and Electronic Engineering,
The Hong Kong Polytechnic University, Hong Kong SAR\\
\textsuperscript{2}School of Transportation and Logistics,
Southwest Jiaotong University, Chengdu, China\\[4pt]
\textsuperscript{*}Corresponding author: \href{mailto:wenbo.fan@polyu.edu.hk}{wenbo.fan@polyu.edu.hk}
}
\par}

\vspace{16pt}

\begin{abstract}
\noindent
Shared mobility systems (e.g., shared cars and ride-hailing services) generate persistent spatial imbalances as vehicles concentrate at popular destinations, leaving trip origins depleted of supply.
Operators incur substantial costs in repositioning empty vehicles, and quantifying the theoretical minimum of this rebalancing distance is practically important. 
Exact computation requires solving a transportation linear program that is challenging at the city scale.

Closed-form approximation models are derived for the minimum rebalancing distance in rectangular service regions.
Parallel derivations are presented for the Manhattan metric (grid road networks) and the Euclidean metric (unconstrained movement).
A scalar spatial imbalance index condenses the full demand pattern into a single interpretable quantity.
Both models share a unified structure: the per-vehicle rebalancing distance scales with the square root of service area, the imbalance index, and a shape factor that depends solely on the aspect ratio.
Calibration and validation against 500 exact LP solutions per metric confirm the area-scaling exponent to within 2\% of the theoretical prediction, across three demand distribution families.

An empirical case study using January 2026 New York City for-hire vehicle trip data across 263 traffic analysis zones confirms that the formula generalizes to real-world, network-constrained demand.
The results equip operators and system designers with solver-free, theoretically grounded tools for benchmarking rebalancing performance and optimizing service, rebalancing frequency, and demand-management interventions.
\end{abstract}
\end{titlepage}

\setcounter{page}{1}
\clearpage

\section{Introduction}
\label{sec:intro}

Shared mobility systems such as carsharing, bike-share, and ride-hailing offer flexible, on-demand transport services that conventional fixed-line public transport cannot provide \citep{shaheen2008carsharing,shaheen2019shared}. Flexibility, however, comes with a persistent operational challenge --- \emph{spatial imbalance} ---  vehicles strand wherever passengers leave them, causing supply and demand to fall progressively out of alignment. To maintain efficient fleet utilization, operators must frequently correct this imbalance by repositioning idle vehicles from surplus to deficit areas, and these movements constitute a
significant share of total operating costs \citep{he_robust_2020,nair_fleet_2011,gottschalg_dynamic_2026}.
This need for active \emph{rebalancing} marks a fundamental operational distinction between shared mobility systems and conventional
fixed-line public transport, whose vehicles return naturally to their starting points along prescribed routes.

Computing the minimum achievable rebalancing cost (e.g., vehicle kilometers traveled, VKT) exactly is highly desirable but is computationally and memory demanding. It requires solving a transportation linear program (LP) with $O(n^3 \log n)$ complexity in the number of
locations~$n$ \citep{meng2025nnsemd}, which can be burdensome for iterative city-scale planning. Consequently, practical studies have largely resorted to numerical optimization restricted to small- and medium-sized instances \citep{raviv2013optimal, guo_robust_2021}, heuristic routing models \citep{schuijbroek2017inventory}, and computationally tractable approximations \citep{braverman2019emptycar, beirigo2021learning}.
A parsimonious model of this repositioning cost would therefore be valuable: it not only informs operational strategies for shared mobility services but also enables insightful comparisons across transport modes under varying conditions.

Analytical models have been established for relevant transportation problems, such as the traveling salesman problem (TSP) and the vehicle routing problem (VRP) in one-to-many settings, where all vehicles start and end at a common depot \citep{beardwood1959shortest, daganzo1984tours, daganzo1984vrp, daganzo1987a, daganzo1987b, estrada2004formulas, lei2024subset}. These results, however, do not extend to rebalancing problems that feature many-to-many settings.
In the latter setting, \cite{daganzo2000asymptotic,daganzo2004bounds} studied stochastic systems in which the net supply has zero mean and whose variance captures the local supply--demand mismatch, leaving observable (or predictable) systematic imbalance unaddressed. The random bipartite matching problem (RBMP) literature \citep{caracciolo2014scaling, caracciolo2015correlation, shen2024rbmp, zhai2026rbmp1d} derives closed-form approximations for expected matching cost in specific regimes, yet does not address shared vehicle rebalancing, which is inherently an Optimal Transport (OT) problem between empirical supply and demand distributions rather than a combinatorial assignment over two fixed point sets.

This paper fills that gap with three contributions. {First, it introduces a scalar \emph{imbalance index} that quantifies structured supply-demand mismatch and accordingly, derives unified closed-form approximation models for minimum rebalancing distance under both Manhattan and Euclidean metrics. Second, it proves a rigorous upper bound on minimum rebalancing distance. Third, the models are validated against 500 exact LP solutions per metric and a January 2026 New York City for-hire vehicle case study.


The remainder of the paper is organized as follows.
Section~\ref{sec:litreview} reviews related literature, and
Section~\ref{sec:method} formalizes the problem and presents the Manhattan
and Euclidean derivations. Section~\ref{sec:numerical} reports calibration, validation,
and robustness results, Section~\ref{sec:case_study} presents the New York
City case study, and Section~\ref{sec:conclusion} concludes with practical
design guidelines.

\section{Literature Review}
\label{sec:litreview}

The development of analytical approximate models for transport costs follows a natural progression: from one-to-many systems anchored at a common depot, such as the TSP and VRP, to many-to-many systems, such as OT and RBMP. This section mirrors that progression, first reviewing the analytical tradition of closed-form routing formulae for the TSP and VRP, then turning to OT and RBMP, and finally identifying the gap that the present paper addresses.

\subsection{Closed-form routing formulae for one-to-many systems}
\label{subsec:tsp_vrp}

The program of expressing routing costs in closed form begins with
\citet{beardwood1959shortest}, who proved that the length of the optimal
traveling salesman tour through $N$ points drawn uniformly at random from a
bounded planar region of area~$R$ satisfies
$L_{\rm TSP} \sim \beta_{\rm TSP}\sqrt{NR}$ almost surely as $N \to \infty$,
with $\beta_{\rm TSP}$ a universal constant.
This result---hereafter BHH---established the core paradigm: a routing cost
can be expressed as a closed-form product of a universal constant, a
density-dependent factor, and an area term, independent of the detailed point
configuration for large~$N$.

\citet{daganzo1984tours} extended the BHH paradigm to account for region shape.
For a rectangular region of dimensions $L \times W$ with $N$ uniformly scattered
points, he showed that elongated regions incur longer tours than compact ones at
equal area and density, and derived an approximate formula capturing this
effect.
The formula transitions between a $\sqrt{R}$ regime for nearly square regions
and a strip-dominated regime, in which tour length
also depends on the aspect ratio (region width in the original formula).
\citet{daganzo1984vrp} derived analogous formulae for the VRP with a central
depot and vehicle capacity~$C$, obtaining
$D_{\rm VRP} \approx 2rN/C + 0.57\sqrt{NR}$, where $r$ is the average
depot--customer distance and the second term retains the BHH form.
\citet{daganzo1987a,daganzo1987b} confirmed that this $\sqrt{R}$-scaling
structure is robust to delivery time windows, and \citet{estrada2004formulas}
validated and extended the VRP formulae to circular and elliptic regions.
More recently, \citet{lei2024subset} generalized the BHH result to
\emph{partial}-coverage tours, deriving a closed-form estimate of the
expected optimal tour length for visiting any subset of $n \leq N$ randomly
distributed points under both Euclidean and rectilinear (Manhattan) metrics,
using a ``trapping effect'' correction to account for the extra distance
incurred when nearby points have already been visited.

Beyond static routing, \citet{stein1978asymptotic} and
\citet{bertsimas1991stochastic} extended the analysis to open and dynamic
settings; the latter showed that the average waiting time in the dynamic
traveling repairman problem exhibits the $\sqrt{R/N}$
local structure in heavy traffic.
Collectively, this body of work demonstrates that closed-form, shape-aware
formulae are both theoretically sound and practically valuable.
Yet all of these results pertain to the TSP or VRP, in which vehicles
visit a set of points along a tour originating from a depot.
This one-to-many structure is fundamentally different from the many-to-many
transport problem, and the formulae derived for it do not transfer directly
to the rebalancing setting.

\subsection{Optimal transport in many-to-many systems}
\label{subsec:otp}

The analytical study of many-to-many transport costs was pioneered by
\citet{daganzo2000asymptotic,daganzo2004bounds}, who showed that for systems driven by stochastic supply--demand fluctuations with mean of zero net supply, the average per-node optimal transport cost in two dimensions $\langle p^* \rangle$
is bounded above by $O(\sigma\delta^{-1/2}\log N)$ as $N \to \infty$,
where $\delta = N/R$ is the demand density and $\sigma$ is the standard
deviation of the net supply at each node. 


The connection between OT theory and shared-mobility
rebalancing was made clear by \citet{spieser2016vehicle}, who formulated the
autonomous mobility-on-demand (AMoD) rebalancing problem as a fluid-flow model under time-varying demand.
They proved that in the special case of stationary demand, the rebalancing problem reduces to an Earth Mover's Distance (EMD) calculation --- equivalently, the first Wasserstein distance $W_1$ \citep{ruschendorf1985wasserstein}.


A parallel mathematical literature derives the expected value of the EMD for
randomly drawn distributions \citep{bourn2020expected,erickson2021generalization,erickson2024cayley}.
These results are elegant but are confined to one-dimensional settings and do
not extend to the two-dimensional spatial geometry of the vehicle rebalancing
problem.
On the computational side, exact evaluation of $W_1$ between two discrete
distributions over $N$ locations requires solving a linear program with $O(N^2)$
variables, incurring $O(N^3 \log N)$ worst-case complexity
\citep{meng2025nnsemd}.
\citet{meng2025nnsemd} addressed this with a nearest-neighbour-search
approximation that achieves 44--135$\times$ speedups over exact EMD solvers.
While valuable for large-scale numerical work, this approach yields no functional
relationship between $W_1$ and external parameters.

\subsection{Random bipartite matching in many-to-many systems}
\label{subsec:rbmp}

The RBMP is a special case of the OT in which $N$ equal-weight supply points are matched one-to-one with $N$ equal-weight demand points, so that the $W_1$ distance reduces to the total distance of that assignment. Whereas the OT literature above characterizes the asymptotic scaling of $W_1$ in full generality, a complementary strand seeks closed-form expressions for the expected matching distance in this more structured setting.

\citet{treleaven2013asymptotically} proved that when supply and demand
locations follow different spatial distributions, the average per-trip
matching distance converges almost surely to the Wasserstein distance between
those two distributions as the fleet size grows, with an explicit bound on
the rate of convergence.
\citet{caracciolo2014scaling} focused on the symmetric case of uniformly
distributed supply and demand, deriving the first closed-form expressions for
how the average optimal matching cost scales with the number of matched
pairs---most notably, that in two dimensions the average per-pair distance
grows as $O\!\bigl((\ln N/N)^{1/2}\bigr)$ with an analytically exact leading
coefficient.
\citet{caracciolo2015correlation} extended this further by characterizing not
just the average cost but the spatial correlations between individual match
lengths, showing these are governed by the same mathematical structure as
heat diffusion on the domain and that this structure is robust across
different choices of distance metric.

Most recent works include \citet{shen2024rbmp} and \citet{zhai2026rbmp1d}.
The former derives closed-form approximations for the expected optimal matching distance in RBMPs without spatial structure, on hyperspheres, and within bounded domains under
general $L^p$ metrics; the latter derives exact and approximate closed-form
formulas for the expected matching distance in one-dimensional balanced and
unbalanced RBMPs.

\subsection{Research gap}
\label{subsec:gap}

The foregoing review reveals a specific gap: {no closed-form approximation relates the minimum rebalancing distance in two-dimensional service regions to its area, aspect ratio, and supply-demand imbalance jointly.} This paper fills that gap with calibrated scaling approximations under both the Manhattan and Euclidean metrics, validated against exact LP solutions.

\section{Methodologies} \label{sec:method}
\subsection{Problem Formulation}
\label{sec:problem}

This section formalizes the problem setting and introduces the key notation, as summarized in Table~\ref{tab:notation}.
Consider a shared mobility system (e.g., one-way shared cars) operating within a \emph{rectangular service region} $\Omega = [0,\,L]\times[0,\,W]$, where $L \geq W > 0$. 
The region has:
\begin{itemize}[leftmargin=2em,itemsep=2pt]
  \item Area: $R = LW$ (km$^2$)
  \item Aspect ratio: $\kappa = L/W \geq 1$
        (by convention the longer dimension is $L$)
  \item Dimensions recoverable as: $L = \sqrt{\kappa R}$,\;
        $W = \sqrt{R/\kappa}$.
\end{itemize}
The rectangle admits a tractable derivation of shape effects \citep{daganzo1984tours, daganzo1984vrp} and also serves as a fundamental building block for more complex geometries (e.g., convex and concave polygons).

Over a rebalancing interval of length $h$ (hours), the system generates
$N$ trips, with $\mathbb{E}[N] = \lambda h R$, where
$\lambda$ (trips$\cdot$km$^{-2}\cdot$h$^{-1}$) is the mean areal
trip-generation rate. 
Each trip has a random \emph{origin} $\mathbf{x}_i$ (vehicle pickup)
drawn from the spatial density $f(\mathbf{x})$, and a random
\emph{destination} $\mathbf{y}_i$ (vehicle drop-off) drawn from the
spatial density $g(\mathbf{x})$.
Both $f$ and $g$ are continuous densities on $\Omega$ normalized so that
$\int_\Omega f = \int_\Omega g = 1$; they are \emph{not} assumed to be
equal or uniform.
The structural mismatch between $f$ and $g$ is the sole source of
systematic rebalancing need.
After all $N$ trips, location $z$ has a net surplus of vehicles if
$\sum_i \mathbf{1}[\mathbf{y}_i \in z] > \sum_i \mathbf{1}[\mathbf{x}_i \in z]$
(more drop-offs than pickups), and a net deficit otherwise.


\begin{definition}[Imbalance Index]
The \emph{imbalance index} is the total variation {distance} between the origin and destination densities ($d_{\rm TV}(f,\,g)$):
\begin{equation}
  I \;=\; d_{\rm TV}(f,\,g)
    \;=\; \frac{1}{2}\int_\Omega \bigl|f(\mathbf{x}) - g(\mathbf{x})\bigr|\dd\mathbf{x}
  \label{eq:I}
\end{equation}
\end{definition}
--- a measure of {how much} they differ in magnitude (not of how far apart they are spatially). By construction, $I \in [0,1]$.
$I=0$ implies perfect balance (no rebalancing needed);
$I=1$ implies origins and destinations are spatially disjoint
(maximum mismatch).
In discrete form, given $M$ shared-bike stations (or shared-car parking lots or ride-hailing service zones) with fractions $\{p_i\}$ (trip origins) and $\{q_i\}$ (trip destinations):
\begin{equation}
  I \;=\; \frac{1}{2}\sum_{i=1}^{M}
          \bigl|q_i - p_i\bigr|.
  \label{eq:Idiscrete}
\end{equation}
Note that $I$ represents the observed or predicted imbalance prior to rebalancing. Just as $\sigma$ governs cost in the stochastic regime \citep{daganzo2004bounds}, $I$ plays the analogous role in the deterministic one but explicitly captures the structured spatial mismatch that is absent in referred work.




After each interval of length $h$, the operator solves a {minimum-cost vehicle relocation problem}:
move empty vehicles from surplus locations to deficit locations at minimum
total distance.
This is exactly a \emph{Kantorovich optimal transport problem}
(equivalently, an Earth Mover's Distance problem) with ground cost
$d(\mathbf{x},\mathbf{y})$, where $d$ is the chosen
travel metric (Manhattan or Euclidean) \citep{villani2003topics}.
Its optimal value is the corresponding \emph{Wasserstein-1 distance},
denoted $W_{1,d}(f,g)$; once the metric is fixed, we write $W_1(f,g)$
for brevity.

The \emph{total rebalancing VKT} over interval $h$ is therefore:
\begin{equation}
  \EVKT \;=\; N \cdot W_1(f,\,g)
            \;=\; \lambda h R \cdot W_1(f,\,g).
  \label{eq:VKT_general}
\end{equation}


\begin{table}[!ht]
\centering
\caption{Table of notation used throughout the paper.}
\label{tab:notation}
\footnotesize\renewcommand{\arraystretch}{1.2}
\begin{tabular}{@{}p{3.2cm}p{5.8cm}p{2.0cm}p{2.0cm}@{}}
\toprule
\textbf{Symbol} & \textbf{Definition} & \textbf{Unit} & \textbf{Range} \\
\midrule
\multicolumn{4}{@{}l}{\textit{Region geometry}} \\[1pt]
$L,\;W$              & Rectangle side lengths ($L\geq W$)       & km            & --          \\
$R = LW$             & Service area                              & km$^2$        & 0.25--50    \\
$\kappa = L/W$       & Aspect ratio                              & --            & $\geq 1$    \\
\midrule
\multicolumn{4}{@{}l}{\textit{Demand process}} \\[1pt]
$\lambda$            & Spatial trip-generation density           & trips$\cdot$km$^{-2}$h$^{-1}$ & 1--500 \\
$h$                  & Rebalancing interval length               & h             & 1--24       \\
$N = \lambda h R$    & Expected trips per interval               & --            & -- \\
$f(\mathbf{x})$      & Normalized origin (pickup) density        & km$^{-2}$     & --          \\
$g(\mathbf{x})$      & Normalized destination (drop-off) density & km$^{-2}$     & --          \\
\midrule
\multicolumn{4}{@{}l}{\textit{Imbalance}} \\[1pt]
$I \in [0,1]$        & Imbalance index $\tfrac{1}{2}\!\int|f-g|$ & --           & --  \\
$\rho = f - g$       & Signed imbalance density                  & km$^{-2}$     & --          \\
\midrule
\multicolumn{4}{@{}l}{\textit{Rebalancing cost}} \\[1pt]
$W_1$                & Wasserstein-1 (mean rebalancing distance) & km            & --          \\
$\mathrm{VKT} = N W_1$ & Total rebalancing vehicle-kilometers    & km            & --          \\
\midrule
\multicolumn{4}{@{}l}{\textit{Model parameters}} \\[1pt]
$g_{\rm M}(\kappa)$  & Shape factor (both metrics), $(\kappa+1)/\sqrt{\kappa}$ & --   & $\geq 2$    \\
$C_{\rm M}$          & Manhattan calibration constant            & --            & --  \\
$C_{\rm E}$          & Euclidean calibration constant & --  & --  \\
$\widehat{C}_{\rm M}$& Estimated $C_{\rm M}$ (median estimator)  & --            & --          \\
$\widehat{C}_{\rm E}$& Estimated $C_{\rm E}$ (median estimator)  & --            & --          \\
\bottomrule
\end{tabular}
\end{table}

\subsection{Derivation: Manhattan (L\texorpdfstring{$^1$}{1}) Metric}
\label{sec:manhattan}

The Manhattan (or $L^1$) metric is
$\|\mathbf{x}-\mathbf{y}\|_1 = |x_1-y_1| + |x_2-y_2|$,
appropriate for vehicles navigating a dense rectilinear (grid) road network.

\subsubsection{Wasserstein-1 Formulation}

The minimum rebalancing cost in the Manhattan metric is:
\begin{equation}
  \W^{\rm M}(f,g)
  = \inf_{\gamma\in\Pi(f,g)}
    \int_{\Omega\times\Omega}
    \bigl(|x_1-y_1|+|x_2-y_2|\bigr)\dd\gamma(\mathbf{x},\mathbf{y}),
  \label{eq:MK_M}
\end{equation}
where $\gamma(\mathbf{x},\mathbf{y})$ is a joint probability measure, indicating what fraction of the total vehicle movements should go from location $\mathbf{x}$ to location $\mathbf{y}$; $\Pi(f,g)$ is the set of all joint probability measures on
$\Omega\times\Omega$ with marginals $f$ and $g$. The $\inf$ over all valid $\gamma\in\Pi(f,g)$ is asking: of all feasible rebalancing plans that satisfy supply and demand balance, which one minimizes total VKT?





\subsubsection{Beckmann's Flow Reformulation}

Instead of tracking specific vehicle itineraries, \cite{beckmann1952continuous} showed that the entire rebalancing
solution can be described as a continuous \emph{vector flux field}
$\mathbf{v}(\mathbf{x}) = (v_x(\mathbf{x}),\, v_y(\mathbf{x}))$ that is locally dependent.
At each location $\mathbf{x}$, the two components tell us the net rate at
which empty vehicles pass through in the east--west ($x$) and north--south
($y$) directions respectively. 
For the flux field to be physically valid, it must satisfy the
{divergence equation} at every interior point:

\begin{equation}
  \underbrace{\frac{\partial v_x}{\partial x}
            + \frac{\partial v_y}{\partial y}}_{\text{net outflow at }\mathbf{x}}
  \;=\;
  \underbrace{\rho(\mathbf{x})}_{\text{local surplus}}
  \qquad \text{in } \Omega,
  \label{eq:continuity}
\end{equation}

where $\rho = f - g$ is the signed imbalance density (positive in surplus
locations, negative in deficit locations), with $\int_\Omega \rho\,\mathrm{d}\mathbf{x}
= 0$ ensuring global supply--demand balance.
No flow is permitted to cross the boundary of the service region:
$\mathbf{v} \cdot \hat{n}\big|_{\partial\Omega} = 0$.

Under the Manhattan ($L^1$) metric, moving a vehicle costs proportionally
to the sum of east--west and north--south distances traveled.
The total rebalancing cost of a plan $\mathbf{v}$ is therefore
$\int_\Omega \bigl(|v_x(\mathbf{x})| + |v_y(\mathbf{x})|\bigr)\,\mathrm{d}\mathbf{x}$,
and we seek the plan that minimizes this over all flows satisfying
\eqref{eq:continuity}:

\begin{equation}
  W_1^{\rm M}(f,g) \;=\; \min_{\mathbf{v}} \int_{\Omega} \|\mathbf{v}(\mathbf{x})\|_1
  \,\mathrm{d}\mathbf{x}
  \quad \text{subject to} \quad
  \nabla \cdot \mathbf{v} = \rho \;\text{ in } \Omega, \quad
  \mathbf{v} \cdot \hat{n}\big|_{\partial\Omega} = 0.
  \label{eq:beckmann_manhattan}
\end{equation}

This is Beckmann's (1952) formulation --- a direct continuous analog of
the classical minimum-cost flow LP that transportation engineers solve on
discretised networks, with $\mathbf{v}(\mathbf{x})$ playing the role of
the arc-flow variable and \eqref{eq:continuity} playing the role of the
node-balance constraints.




\subsubsection{Decomposing in the Manhattan Metric}

\noindent
On a grid road network, a vehicle can only move
east--west \emph{or} north--south at any moment---never diagonally.
This means the total driving distance is simply
\[
  \text{total distance} \;=\; \underbrace{\text{east--west distance}}_{\text{along }x}
                            \;+\; \underbrace{\text{north--south distance}}_{\text{along }y}.
\]
Because the two directions \emph{add up}, the 2-D rebalancing cost can be {approximated}
by solving two separate 1-D problems---one per axis---and summing
the results.
This axis-projection approximation provides a \emph{lower bound} on the
true 2-D Manhattan cost.
The bound is tight when $\rho(x,y)$ is approximately \emph{additively
separable}; it is loose
when the imbalance is concentrated in ``diagonal'' modes that project
to zero on both axes simultaneously (e.g.\ a checkerboard pattern of
alternating surplus and deficit locations).
In practice, the approximation quality depends on the spatial geometry of $\rho$ and should be assessed empirically for a given system; we treat the axis-projection result as a tractable lower bound and quantify the approximation error in Sections \ref{sec:numerical} and \ref{sec:case_study}.

\noindent
{For each 1-D sub-problem,} think of the service region  $\Omega$ as a grid of city blocks.
At each location $\mathbf{x}$, define the \emph{signed imbalance density}
$\rho(\mathbf{x}) = f(\mathbf{x})-g(\mathbf{x})$, where $f$ is the origin
(pickup) density and $g$ is the destination (drop-off) density.
To decompose the 2-D problem:

\begin{enumerate}
  \item {East--west sub-problem.}
    Collapse the entire region onto the $x$-axis by summing (integrating)
    $\rho$ over all north--south positions at each $x$-coordinate.
    This gives a 1-D imbalance profile
    \[
      \rho_x(x) \;=\; \int_0^W \rho(x,\,y)\,\dd y
      \qquad \text{(net surplus at each east--west position)}.
    \]

  \item {North--south sub-problem.}
    Similarly, collapse onto the $y$-axis:
    \[
      \rho_y(y) \;=\; \int_0^L \rho(x,\,y)\,\dd x
      \qquad \text{(net surplus at each north--south position)}.
    \]
\end{enumerate}

\bigskip
\noindent
As a scaling approximation, we project the 2-D problem onto its two
coordinate axes.
The sum of the resulting 1-D transport costs is a lower bound on the
full 2-D Manhattan cost; the gap is small when $\rho$ is approximately
additively separable across axes and can be large when $\rho$ has
strong diagonal structure (see Remark 2 below).
We write:
\begin{subequations} \label{eq:decouple}
\begin{align}
  W_1^{\rm M}(f,\,g)
  \;\approx\;
    \min_{v_x}\;\int_0^L \bigl|v_x(x)\bigr|\,\dd x
  \;+\;
    \min_{v_y}\;\int_0^W \bigl|v_y(y)\bigr|\,\dd y
\end{align}
subject to:
\begin{align}
    \frac{\dd v_x}{\dd x} = \rho_x(x),\\
    \frac{\dd v_y}{\dd y} = \rho_y(y).
\end{align}
\end{subequations}

\bigskip
\noindent
Both sub-problems have the same structure:

\begin{itemize}
  \item $v_x(x)$ is the \emph{net rate of empty vehicles passing through}
        position $x$ in the east--west direction
        (positive $=$ moving east; negative $=$ moving west).
  \item The constraint $\frac{\dd v_x}{\dd x} = \rho_x(x)$ is a {flow conservation
        rule}: the rate at which flow ``builds up'' at position $x$ must
        equal the local surplus $\rho_x(x)$.
        Vehicles piling up in a surplus location must be pushed out; deficit
        locations pull vehicles in.
  \item $\int_0^L |v_x(x)|\,\dd x$ accumulates the {absolute flow} at
        every position---this is the total east--west vehicle distance to
        be minimized.
  \item The north--south term works identically but along $y$.
\end{itemize}


\subsubsection{One-Dimensional Subproblem} \label{sec:scaling_manhattan}

Along the $x$-axis of length $L$, with marginal imbalance profile
$\rho_x(s) = \int_0^W \rho(s,y)\,\dd y$
(the net surplus per unit length at position $s\in[0,L]$),
the optimal flow satisfies:
\[
  \frac{\dd v_x}{\dd s} = \rho_x(s),\quad v_x(0)=v_x(L)=0
  \implies
  v_x^*(s) = \int_0^s \rho_x(u)\,\dd u.
\]
The east--west contribution to the per-vehicle rebalancing distance is:
\[
  \text{Cost}_x = \int_0^L |v_x^*(s)|\dd s = \int_0^L \left| \int_0^s \rho_x(u)\dd u \right|\dd s.
\]
Note that $v_x^*$ already carries the full-width flow (since $\rho_x$
integrates $\rho$ over $y$); no additional factor of $W$ is needed.
This $\mathrm{Cost}_x$ is the $x$-component of $W_1^{\rm M}$ in
Eq.~\eqref{eq:decouple}.

\begin{proposition}[Upper bound on the 1-D rebalancing cost]
\label{prop:upper_bound}
For any continuous densities $f,g$ on $\Omega=[0,L]\times[0,W]$ with
imbalance index $I = \tfrac{1}{2}\int_\Omega|f-g|$, the east--west
component of $W_1^{\rm M}$ satisfies
\begin{equation}\label{eq:costx_upper}
  \mathrm{Cost}_x \;\leq\; I\cdot L.
\end{equation}
\end{proposition}
\begin{proof}
Because $v_x^*(s)=\int_0^s\rho_x(u)\,\dd u$ and $\rho_x =
\int_0^W\rho\,\dd y$, Jensen's inequality gives
$|v_x^*(s)| = |\int_0^s\rho_x\,\dd u|
\leq \int_0^s|\rho_x(u)|\,\dd u
\leq \int_0^L|\rho_x(u)|\,\dd u$.
The last integral equals $\int_0^L|\int_0^W\rho\,\dd y|\,\dd x
\leq \int_\Omega|\rho|\,\dd\mathbf{x} = 2I$, so
$|v_x^*(s)|\leq 2I$ for every~$s$.
A tighter bound follows from the constraint that $v_x^*$ starts and ends
at zero: the positive and negative parts of $\rho_x$ each integrate to at
most $\int_{[\rho_x>0]}\rho_x\,\dd x \leq \int_\Omega[\rho]^+\,\dd\mathbf{x} = I$,
hence $|v_x^*(s)|\leq I$ for all~$s$.
Therefore
$\mathrm{Cost}_x = \int_0^L|v_x^*|\,\dd s \leq I\cdot L$.
\end{proof}

The upper bound is complemented by a structural observation.
Because the flow equation $\dd v_x/\dd s = \rho_x$ is linear,
scaling the imbalance density $\rho\mapsto\alpha\rho$ (which maps
$I\mapsto\alpha I$) scales $v_x^*\mapsto\alpha v_x^*$ and hence
$\mathrm{Cost}_x\mapsto\alpha\,\mathrm{Cost}_x$.
That is, $\mathrm{Cost}_x$ is \emph{1-homogeneous} in~$I$.
Dimensional analysis then fixes the remaining dependence: $\mathrm{Cost}_x$
has the unit of length, and the only length scale in the 1-D sub-problem
is~$L$, so
\begin{equation}\label{eq:costx}
  \mathrm{Cost}_x \;=\; C(\text{shape})\;\cdot I\cdot L,
  \tag{\ref{eq:costx_upper}$'$}
\end{equation}
where $C(\text{shape})\in(0,1]$ is a dimensionless constant that depends
on the shape of $\rho_x$ but not on its amplitude or on $L$.
The upper bound (Proposition~\ref{prop:upper_bound}) confirms $C\leq 1$.

\paragraph{Scaling argument for a representative profile.}
To fix the order-of-magnitude constant, we evaluate $\mathrm{Cost}_x$
under three simplifying assumptions stated together:
\begin{enumerate}[label=\textbf{A\arabic*},leftmargin=2em]
  \item \emph{Uniform spreading.}
        $\rho(x,y)$ has a roughly constant sign across $y$ at each~$x$,
        so $|\int_0^W\!\rho\,\dd y|\approx\int_0^W\!|\rho|\,\dd y$.
  \item \emph{Symmetric profile.}
        The sign change of $\rho_x$ occurs near the midpoint $x\approx L/2$.
  \item \emph{Triangular shape.}
        The flow profile $|v_x^*(x)|$ is approximated by a triangle of
        base~$L$.
\end{enumerate}
Under A1, integrating the average 2-D imbalance density $2I/R$ over
width~$W$ gives $|\rho_x|\sim 2I\,W/R = 2I/L$.
The running integral $v_x^*=\int_0^x\rho_x\,\dd s$ then grows from
zero to a peak magnitude of order $(I/L)\cdot(L/2)=I/2$ at the midpoint
(A2), and returns to zero at $x=L$ by global balance.
Approximating $|v_x^*|$ as a triangle of base $L$ and height $I/2$
(A3) yields
\begin{equation}\label{eq:costx_derived}
  \mathrm{Cost}_x
  \;\sim\; \tfrac{1}{2}\cdot L\cdot\tfrac{I}{2}
  \;=\; \tfrac{IL}{4}
  \;\sim\; I\cdot L,
\end{equation}
confirming the scaling law in Eq.~\eqref{eq:costx}.
The numerical prefactor $\tfrac{1}{4}$ and the order-one errors
introduced by A1--A3 are absorbed into the calibration constant~$C_{\rm M}$.

\begin{corollary}[Upper bound on the north--south rebalancing cost]
\label{cor:costy_upper}
By symmetry (replace $L\leftrightarrow W$ throughout
Proposition~\ref{prop:upper_bound}),
\begin{equation}\label{eq:costy_upper}
  \mathrm{Cost}_y \;\leq\; I\cdot W.
\end{equation}
\end{corollary}

\paragraph{Scaling argument for $\mathrm{Cost}_y$.}
The north--south sub-problem is structurally identical to the east--west
one: replace $L\leftrightarrow W$ throughout the scaling argument above.
The same 1-homogeneity and dimensional analysis give
$\mathrm{Cost}_y = C(\text{shape})\cdot I\cdot W$, and Assumptions A1--A3
yield
\begin{equation}\label{eq:costy}
  \mathrm{Cost}_y \;\sim\; \tfrac{IW}{4} \;\sim\; I\cdot W,
\end{equation}
confirming that the north--south rebalancing cost grows with region
width~$W$ for the same reasons that $\mathrm{Cost}_x$ grows with~$L$. 
The same caveats apply as for $\mathrm{Cost}_x$: this is a heuristic
scaling estimate under Assumptions A1--A3, not a rigorous bound on
$W_1^{\rm M}$ for general $f,g$.

\subsubsection{Final Manhattan Model} \label{sec:FinalManhattanModel}

Inserting the scaling results~\eqref{eq:costx_derived} and~\eqref{eq:costy}
into the decomposition~\eqref{eq:decouple}:
\begin{equation}
  \W^{\rm M} \;\approx\; C_{\rm M} \cdot I \cdot (L + W),
  \label{eq:W1M_raw}
\end{equation}
where $C_{\rm M}>0$ absorbs the proportionality constant from the
scaling argument.

Substituting $L = \sqrt{\kappa R}$ and $W = \sqrt{R/\kappa}$:
\begin{equation}
  L + W = \sqrt{\kappa R} + \sqrt{\frac{R}{\kappa}}
        = \sqrt{R}\left(\sqrt{\kappa} + \frac{1}{\sqrt{\kappa}}\right)
        = \sqrt{R}\cdot\frac{\kappa+1}{\sqrt{\kappa}}
        = \sqrt{R}\cdot\gM.
  \label{eq:LplusW}
\end{equation}
\begin{remark}[Opposing approximation errors and robustness of the functional form]
The derivation involves two steps with opposing biases.
First, the axis decomposition (Eq.~\ref{eq:decouple}) yields a lower bound on
$W_1^{\rm M}$ by discarding diagonal interactions; the gap can be substantial
when $\rho$ has strong diagonal structure.
Second, Assumption~A1 replaces $|\int_0^W\rho\,\dd y|$ with
$\int_0^W|\rho|\,\dd y$, inflating the estimated marginal imbalance.
The functional form $W_1 \propto I\cdot(L+W)$ is nevertheless robust: it is
supported independently by Proposition~\ref{prop:upper_bound}, which
guarantees $\mathrm{Cost}_x\leq I\cdot L$ (and symmetrically
$\mathrm{Cost}_y\leq I\cdot W$), so the true $W_1$
is sandwiched between the decomposition lower bound and $I\cdot(L+W)$ regardless
of demand structure.
The calibration constant $C$ absorbs the net effect of these offsetting errors
and varies with the demand pattern.
\label{remark_errors}
\end{remark} 

\begin{remark}[Directional demand and the scope of $g_{\rm M}$]
\label{rem:gM_scope}
The shape factor $g_{\rm M}(\kappa)=(L+W)/\sqrt{R}$ arises from combining
$\mathrm{Cost}_x \propto I\cdot L$ and $\mathrm{Cost}_y \propto I\cdot W$
with equal proportionality constants, implicitly assuming that the demand
imbalance generates comparable east--west and north--south rebalancing flows.
For strongly directional demand---e.g., origins concentrated in the western
half of $\Omega$ and destinations in the eastern half---the north--south
rebalancing cost vanishes ($\mathrm{Cost}_y=0$ exactly) and the exact
per-vehicle cost is $W_1^{\rm M}=IL/2$, independent of~$W$.
This motivates an anisotropic refinement via axis-separate approximations as given in Appendix~\ref{appx_anisotropic}. For the spirit of simplicity and parsimony, we focus on the isotropic form in the main text, but practitioners facing strongly directional demand should consider the anisotropic refinement and calibrate axis-dependent parameters separately from historical data.
\end{remark}


Substituting \eqref{eq:LplusW} into \eqref{eq:W1M_raw} gives: 

\begin{equation}
  \W^{\rm M} \;\approx\; C_{\rm M}\cdot I\cdot\sqrt{R}\cdot\gM,
  \qquad
  g_{\rm M}(\kappa) = \frac{\kappa+1}{\sqrt{\kappa}}.
  \label{eq:W1M}
\end{equation}
Multiplying
by $N = \lambda h R$ yields the Manhattan (L\texorpdfstring{$^1$}{1}) model:
\begin{equation}
  \mathbb{E}[VKT^{\rm M}]
  \;=\; N\cdot\W^{\rm M}
  \;\approx\; C_{\rm M}\cdot I\cdot(\lambda h)\cdot R^{3/2}\cdot\gM.
  \label{eq:VKTM}
\end{equation}
\subsection{Derivation: Euclidean (L\texorpdfstring{$^2$}{2}) Metric}
\label{sec:euclidean}

The Euclidean (or $L^2$) metric is
$\|\mathbf{x}-\mathbf{y}\|_2 = \sqrt{(x_1-y_1)^2+(x_2-y_2)^2}$,
appropriate when vehicles can travel in any direction (e.g.\ floating
zones, drone delivery, or as a theoretical lower bound). 
Beckmann's theorem applies equally to the Euclidean metric with the
$L^1$ objective replaced by $L^2$:
\[
  \W^{\rm E}(f,g)
  = \min_{\mathbf{v}}
    \int_\Omega \|\mathbf{v}(\mathbf{x})\|_2\,\dd\mathbf{x}
  \quad\text{subject to}\quad
  \nabla\cdot\mathbf{v} = \rho\;\text{ in }\Omega,\quad
  \mathbf{v}\cdot\hat{n}\big|_{\partial\Omega}=0.
  \label{eq:W1E_Beckmann}
\]



\smallskip\noindent
Although Euclidean vehicles can travel diagonally, the total
rebalancing cost is dominated by the two axis-aligned components of
the imbalance --- one along $L$ and one along $W$.
A Fourier decomposition of the imbalance field $\rho$ on the rectangle
$\Omega$ (Appendix~\ref{appx_proof}) shows that the
east--west component contributes a cost proportional to $I\cdot L$
and the north--south component contributes $I\cdot W$.

\smallskip\noindent
Summing both axis contributions gives
$\W^{\rm E} \approx C_{\rm E}\cdot I\cdot(L+W)$.
Since $L+W = \sqrt{R}\,g_{\rm M}(\kappa)$ (Eq.~\ref{eq:LplusW}),
the Euclidean model inherits the \emph{same} shape factor $g_{\rm M}$
as the Manhattan model:
\begin{equation}
  \W^{\rm E} \;\approx\; C_{\rm E}\cdot I\cdot\sqrt{R}\cdot g_{\rm M}(\kappa).
  \label{eq:W1E}
\end{equation}

\smallskip\noindent
where the only difference from the Manhattan model is the calibration constant.

Multiplying by the number of trips $N = \lambda h R$ gives the \emph{Euclidean (L$^2$) model}

\begin{equation}
  \mathbb{E}[VKT^{\rm E}]
  \;=\; N\cdot\W^{\rm E}
  \;\approx\; C_{\rm E}\cdot I\cdot(\lambda h)\cdot R^{3/2} \cdot \gM.
  \label{eq:VKTE}
\end{equation}

\subsection{Unified Model and Upper Bound on 2-D Rebalancing Cost}
\label{sec:summary}

Both models predict total rebalancing VKT as a product of imbalance
index $I$, trip volume $\lambda h$, and a region-size factor
$R^{3/2}$, modulated by a metric-specific shape factor:
\begin{equation}
  \mathbb{E}[VKT]
  \;\approx\; C\cdot I\cdot(\lambda h)\cdot R^{3/2}\cdot g_{\rm M}(\kappa),
  \label{eq:unified}
\end{equation}
which implies that (i) doubling the service area doubles the number of vehicles to relocate
($\propto R$) \emph{and} increases the average relocation distance
($\propto\sqrt{R}$), giving $VKT\propto R\cdot\sqrt{R} = R^{3/2}$; and (ii) both metrics share the same shape factor
$g_{\rm M}(\kappa)=(\kappa+1)/\sqrt{\kappa}$, which is minimized at
$\kappa=1$, so a square region minimizes rebalancing cost regardless of
metric.




  
   Our model reveals a contrast to the shape-independent finding in \cite{daganzo2004bounds}: once structural imbalance is incorporated, the cost acquires an explicit shape factor $g_{\rm M}(\kappa)$, because imbalance flows span the entire service region rather than merely local neighborhoods, making the geometry of the region consequential.

It is worth noting that Eq.~\eqref{eq:unified} is an approximation with the constant $C$ to be calibrated for certain demand patterns, rather than a rigorous result on $\mathbb{E}[VKT]$ and $W_1$ for general $f,g$.
The following proposition establishes a rigorous upper bound for $W_1$.

\begin{proposition}\label{prop:diameter_bound}
For any demand distributions $f,g$ on a rectangle $\Omega=[0,L]\times[0,W]$ with area $R=LW$, aspect ratio $\kappa=L/W$, and imbalance index $I=\tfrac{1}{2}\int_\Omega|f-g|$, the Wasserstein-1 distance satisfies
\begin{equation}\label{eq:diameter_bound}
    W_1(f,g) \;\leq\; I \cdot \sqrt{R} \cdot g_{\rm M}(\kappa)
\end{equation}
under both the Manhattan and Euclidean metrics.
\end{proposition}
\begin{proof}
Let $[\rho]^{+}=\max(\rho,0)$ and $[\rho]^{-}=\max(-\rho,0)$ denote the surplus and deficit parts of $\rho=f-g$, so that $\int_\Omega[\rho]^{+}=\int_\Omega[\rho]^{-}=I$.
For any transport plan $\gamma$ of $[\rho]^{+}$ and $[\rho]^{-}$,
\[
    \int_{\Omega\times\Omega} d(x,y)\,\dd\gamma(x,y)
    \;\leq\; \mathrm{diam}(\Omega)\cdot\int_{\Omega\times\Omega}\dd\gamma(x,y)
    \;=\; \mathrm{diam}(\Omega)\cdot I.
\]
Taking the infimum over all transport plans gives $W_1(f,g)\leq\mathrm{diam}(\Omega)\cdot I$. 
Under the Manhattan metric, the diameter is achieved at antipodal corners:
$\mathrm{diam}^{\rm M}(\Omega)=L+W=\sqrt{\kappa R}+\sqrt{R/\kappa}=\sqrt{R}\,g_{\rm M}(\kappa)$.
Under the Euclidean metric, $\mathrm{diam}^{\rm E}(\Omega)=\sqrt{L^2+W^2}\leq L+W=\sqrt{R}\,g_{\rm M}(\kappa)$,
since $(L+W)^2=L^2+2LW+W^2\geq L^2+W^2$.
In both cases, $W_1(f,g)\leq I\cdot\sqrt{R}\cdot g_{\rm M}(\kappa)$.
\end{proof}

Eq.~\eqref{eq:unified} thus lies within $I\cdot(\lambda h)\cdot R^{3/2}\cdot g_{\rm M}(\kappa)$ by a factor $C \in (0,1]$.

\section{Numerical Study}
\label{sec:numerical}

\subsection{Study Design} \label{sec:study_design}

We generate random rebalancing instances by independently sampling:
(i) area $R\in\{0.25,\,0.5,\,1,\,2,\,4,\,8,\,16,\,32\}$\,km$^2$;
(ii) aspect ratio $\kappa\in\{1,\,1.5,\,2,\,3,\,4,\,6,\,9\}$.
From these, $L=\sqrt{\kappa R}$ and $W=\sqrt{R/\kappa}$ are computed.
For each instance, $M$ trip origins (vehicle pickup locations) and $M$
trip destinations (vehicle drop-off locations) are drawn independently and
uniformly over $[0,L]\times[0,W]$, where $M$ is drawn from a
$\text{Poisson}(22)$ distribution truncated to $[6,\,32]$ to ensure
LP tractability.
The exact Wasserstein-1 distance is computed by solving the surplus--deficit transportation LP on the same $8\times 8$ grid ($N_z=64$ cells) used to compute $I$ \eqref{eq:Idiscrete}. This same partitioning is used across scenarios to ensure computational efficiency and consistent LP solutions. In practical rebalancing tasks, the service zones (or stations) are given as input; we do not attempt here to determine the optimal partitioning (or to locate shared-bike stations and shared-car parking lots).

Using the trip-origin fractions $\{p_k\}$ and trip-destination fractions
$\{q_k\}$ defined in \eqref{eq:Idiscrete}, the net surplus and deficit at
each cell are
\begin{equation}
  s_k = q_k - p_k > 0,\; k\in\mathcal{S}=\{k:q_k>p_k\}, \,
  d_k = p_k - q_k > 0,\; k\in\mathcal{D}=\{k:p_k>q_k\},
  \label{eq:surplus_deficit}
\end{equation}
where $\mathcal{S}$ (surplus cells, drop-offs exceed pickups) and
$\mathcal{D}$ (deficit cells, pickups exceed drop-offs) satisfy
$\sum_{k\in\mathcal{S}}s_k = \sum_{k\in\mathcal{D}}d_k = I$
by construction.
Let $c_{kl}=d(\mathbf{x}_k,\mathbf{x}_l)$ be the inter-centroid travel
distance and $\gamma_{kl}\geq 0$ the fraction of vehicles relocated from
cell~$k$ to cell~$l$.
The LP minimizes total repositioning distance:
\begin{align}
  \min_{\gamma_{kl}\geq 0} \quad
    & \sum_{k\in\mathcal{S}}\sum_{l\in\mathcal{D}} c_{kl}\,\gamma_{kl}
    \notag\\
  \text{s.t.}\quad
    & \sum_{l\in\mathcal{D}} \gamma_{kl} = s_k,\quad k\in\mathcal{S}
      \quad\text{(surplus absorbed)},
      \notag\\
    & \sum_{k\in\mathcal{S}} \gamma_{kl} = d_l,\quad l\in\mathcal{D}
      \quad\text{(deficit satisfied)}.
  \label{eq:transport_LP}
\end{align}
The optimal value $W_{1,\mathrm{exact}}
= \sum_{k,l} c_{kl}\,\gamma_{kl}^{*}$
approximates the Wasserstein-1 distance $W_1(f,g)$
\citep{villani2003topics}, converging to it exactly as $N_z\to\infty$.
The LP is solved by the HiGHS solver via SciPy.

For calibration, we define the per-instance calibration ratio:
\[
  C_{\mathrm{inst}} = \frac{W_{1,\mathrm{exact}}}{I\cdot\sqrt{R}\cdot g_{\rm M}(\kappa)},
\]
where $g_{\rm M}(\kappa) = (\kappa+1)/\sqrt{\kappa}$ for both metrics.
Three estimators of $C$ are compared:
\begin{enumerate}[leftmargin=2em,itemsep=3pt]
  \item {Robust median.}
    $\hat{C}_{\rm med}=\text{median}\{C_{\mathrm{inst}}\}$ ---
    resistant to outliers and makes no distributional assumption.
  \item {Constrained log-ordinary-least-squares (log-OLS).}
    $\hat{C}_{\rm log}=\exp\!\left(\text{mean}\{\log C_{\mathrm{inst}}\}\right)$
    --- the geometric-mean estimator, equivalent to OLS in log-space with
    unit exponents fixed.
  \item {Free log-OLS.}
    Jointly estimates $C$ and exponents $(\alpha,\beta,\gamma)$ from:
    \[
      \log W_1 = \log C + \alpha\log I
               + \beta\log\sqrt{R} + \gamma\log g_{\rm M}(\kappa).
    \]
    Theory predicts $\alpha=\beta=\gamma=1$.
\end{enumerate}

\subsection{Calibration Results}

Table~\ref{tab:calibration} presents the calibration results for both metrics.
It is observed that the two medians ($\hat{C}_{\rm M}=0.1440$, $\hat{C}_{\rm E}=0.1189$) satisfy
$\hat{C}_{\rm E}/\hat{C}_{\rm M} = 0.825$, within 6\% of the theoretical prediction
$\pi/4 \approx 0.785$ from the isotropic-direction heuristic in Appendix~\ref{appx_proof}.
The small gap reflects finite-sample bias from the discrete LP solver and
the isotropy assumption underlying the Fourier mode analysis.
\begin{table}[htbp]
\centering\small
\caption{Calibration of the constant $C$ ($n_{\rm sim}=500$ per metric).}
\label{tab:calibration}
\renewcommand{\arraystretch}{1.35}
\begin{tabular}{lcc}
\toprule
Estimator & Manhattan ($g_{\rm M}$) & Euclidean ($g_{\rm M}$) \\
\midrule
Robust median $\hat{C}_{\rm med}$     & 0.1440 & 0.1189 \\
95\% confidence interval              & [0.1408,\,0.1473] & [0.1156,\,0.1221] \\
Constrained log-OLS $\hat{C}_{\rm log}$ & 0.1459 & 0.1203 \\
Free log-OLS $\hat{C}_{\rm ols}$       & 0.1818 & 0.1380 \\
Standard error (median)              & 0.00166 & 0.00165 \\
$R^2_{\rm log}$ (free OLS)           & 0.922   & 0.897 \\
\bottomrule
\end{tabular}
\end{table}



\subsection{Prediction Accuracy}

Figure~\ref{fig:parity} shows the predicted $W_1$ versus the LP-exact $W_1$ for all
500 instances per metric, where predictions use the calibrated robust median
$\hat{C}_{\rm M}=0.1440$ for the Manhattan model and
$\hat{C}_{\rm E}=0.1189$ for the Euclidean model.
The model tracks the 1:1 reference closely across more than three orders
of magnitude in $W_1$.
Table~\ref{tab:errors} reports the full prediction error suite.
The Euclidean model exhibits somewhat higher mean absolute percentage error (MAPE) (25.2\%) compared to
Manhattan (18.9\%), reflecting the additional approximation involved in
the Fourier mode analysis (Appendix~\ref{appx_proof}) relative to the
axis-decoupling approximation used for Manhattan.
Both models are unbiased (mean bias error, MBE, near zero).

\begin{figure}[htbp]
  \centering
  \includegraphics[width=\linewidth]{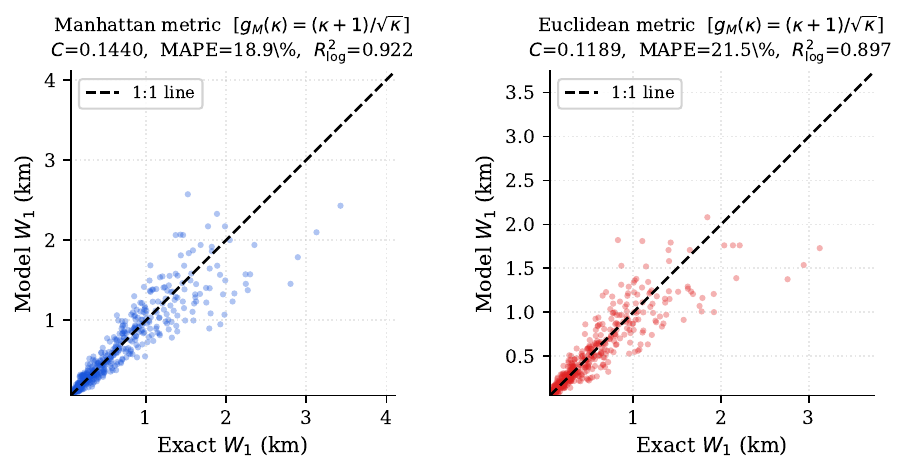}
  \caption{Parity plots of model-predicted versus LP-exact per-vehicle
    distance $W_1$ for the Manhattan (left) and Euclidean (right) metrics.
    Each point represents one randomly generated rebalancing instance ($n_{\rm sim}=500$).
    Dashed line: 1:1 reference.}
  \label{fig:parity}
\end{figure}

\begin{table}[htbp]
\centering\small
\caption{Prediction error metrics on 500 instances per metric
(calibrated $C$ from the robust median estimator).}
\label{tab:errors}
\renewcommand{\arraystretch}{1.35}
\begin{tabular}{lccc}
\toprule
Error metric & Manhattan & Euclidean & Notes \\
\midrule
Mean absolute error (km)    & 0.130 & 0.143 & Absolute scale \\
Root mean square error, RMSE (km) & 0.232 & 0.251 & Penalises outliers \\
MAPE (\%)                   & 18.9  & 25.2  & Relative error \\
$R^2$ (linear scale)        & 0.820 & 0.796 & Misleading for power-law \\
$R^2$ (log scale)$^\dagger$ & 0.922 & 0.875 & Appropriate metric \\
Mean bias error, MBE (km)        & 0.024 & 0.029 & Near-zero: unbiased \\
Median absolute error (km)  & 0.059 & 0.067 & Robust measure \\
p95 absolute percentage error, APE (\%) & 49.8  & 64.1  & Tail performance \\
\bottomrule
\multicolumn{4}{p{12cm}}{\footnotesize
$^\dagger$ Log-scale $R^2$ is the appropriate goodness-of-fit metric for
multiplicative power-law models.}
\end{tabular}
\end{table}



\subsection{Verification}

Table~\ref{tab:exponents} compares the free log-OLS fitted exponents with the theoretical prediction of 1.0. The most important result is $\beta\approx 1.02$ for both metrics, validating the $\sqrt{R} \;\W$ and the $R^{3/2}$ total VKT scaling as a robust quantitative law.
The exponents $\alpha<1$ and $\gamma<1$ are discussed below.

\begin{table}[!htbp]
\centering\small
\caption{Fitted log-OLS exponents vs.\ theoretical predictions.}
\label{tab:exponents}
\renewcommand{\arraystretch}{1.35}
\begin{tabular}{lcccc}
\toprule
Parameter & Theory & Manhattan & Euclidean & Interpretation \\
\midrule
$\alpha$ (exponent of $I$)       & 1.0 & 0.799 & 0.671 & Sub-linear (spatial diffusion) \\
$\beta$  (exponent of $\sqrt{R}$) & 1.0 & 1.019 & 1.013 & Confirmed to $<$2\% \\
$\gamma$ (exponent of $g_{\rm M}(\kappa)$) & 1.0 & 0.666 & 0.730 & Sub-linear (boundary effects) \\
\bottomrule
\end{tabular}
\end{table}



The sub-linear exponent on $I$ ($\alpha < 1$) arises because the model
treats $I$ as a pure scaling factor, whereas in practice higher $I$
correlates with more spatially concentrated imbalance that forces
longer individual trips.
The cost therefore grows slower than linearly in $I$, giving
$\alpha_{\rm M}\approx 0.80$ and $\alpha_{\rm E}\approx 0.67$.
The Euclidean metric is more affected because it exploits local
cancellations more efficiently when imbalance is diffuse.

Similarly, the sub-linear exponent on $g_{\rm M}(\kappa)$ ($\gamma < 1$) reflects
the fact that very elongated regions ($\kappa\gg 1$) push transport into
one dimension along the long axis, where Manhattan and Euclidean
distances coincide.
The actual cost then grows slower with $\kappa$ than the 2-D formula
predicts, giving $\gamma_{\rm M}\approx 0.67$ and
$\gamma_{\rm E}\approx 0.73$.
The Manhattan exponent is lower because elongation drives the
rebalancing problem toward purely one-dimensional transport, where
Manhattan and Euclidean distances coincide. 
In other words, Euclidean costs must therefore grow faster with $\kappa$ to close this
gap, yielding a larger fitted exponent $\gamma_{\rm E}>\gamma_{\rm M}$.
\subsection{Robustness across demand distributions}
\label{sec:robustness}

The theoretical derivation places no restriction on the spatial densities
$f$ and $g$ beyond continuity and normalization; the calibration in
Section~\ref{sec:study_design} uses independently and uniformly
distributed origins and destinations as the baseline demand pattern.
We test robustness across three demand distribution families, as shown in Figure \ref{fig:demand_pattern}, using a
train/test split (60 instances per split per condition).
\begin{itemize}[leftmargin=2em,itemsep=2pt]
  \item {Uniform}: origins and destinations i.i.d.\
        uniform on $\Omega$. Serves as the baseline calibration case.
  \item {Spatially clustered}: origins and destinations each drawn
        from a $k$-component Gaussian mixture ($k\approx N/5$).
        Mimics dense activity centers.
  \item {Directional asymmetric}: origins concentrated in the
        western half of $\Omega$, destinations in the eastern half.
        Mimics a dominant commute flow.
\end{itemize}
\begin{figure}[!ht]
    \centering
    \includegraphics[width=1\linewidth]{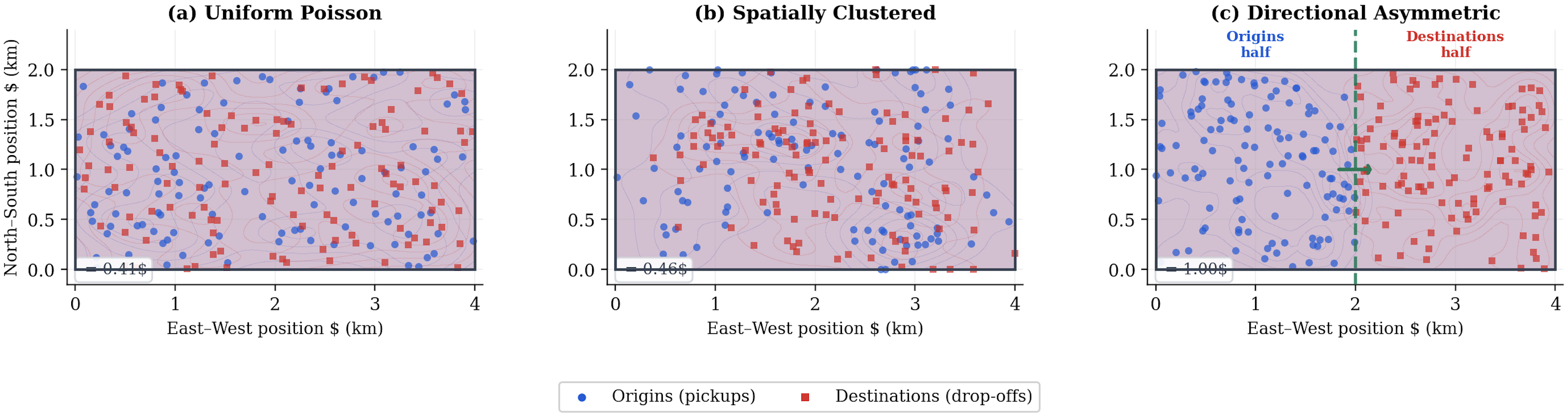}
    \caption{Illustration of three demand distribution families.}
    \label{fig:demand_pattern}
\end{figure}
Table~\ref{tab:robustness} summarises results.
Two findings are notable.
First, the {model structure is robust}: $R^2_{\rm log}>0.85$ for all
demand types and both metrics, confirming that $I$, $\sqrt{R}$, and $g_{\rm M}(\kappa)$
remain the correct explanatory factors regardless of the spatial demand pattern.
Second, {$C$ depends substantially on demand type}: it ranges from
$\approx 0.12$--$0.15$ (uniform) to $\approx 0.38$--$0.40$ (asymmetric),
roughly a three-fold variation.
\paragraph{Interpretation of $C$.}
    $C$ is essentially measuring how concentrated the imbalance is, and how efficiently the transport plan can exploit short-range cancellations. Uniformly distributed demand has many local cancellations (nearby origins and destinations neutralise each other), keeping $C$ small. Asymmetric demand has almost no local cancellations — every vehicle must cross the midline — driving $C$ to its maximum but still within the upper bound per Proposition~\ref{prop:upper_bound}.

\begin{table}[!htbp]
\centering\small
\caption{Robustness: calibration constant $C$ and out-of-sample error by
  demand type (train/test split, 60 instances each; $C$ calibrated within
  each demand family).}
\label{tab:robustness}
\renewcommand{\arraystretch}{1.35}
\begin{tabular}{llccc}
\toprule
Demand type & Metric & $\hat{C}_{\rm med}$ & MAPE (\%) & $R^2_{\rm log}$ \\
\midrule
\multirow{2}{*}{Uniform}
  & Manhattan & 0.1469 & 20.3 & 0.925 \\
  & Euclidean & 0.1180 & 19.5 & 0.918 \\
\midrule
\multirow{2}{*}{Spatially clustered}
  & Manhattan & 0.2129 & 26.3 & 0.855 \\
  & Euclidean & 0.1688 & 26.1 & 0.878 \\
\midrule
\multirow{2}{*}{Directional asymmetric}
  & Manhattan & 0.3967 & 11.9 & 0.978 \\
  & Euclidean & 0.3771 & 19.0 & 0.940 \\
\bottomrule
\end{tabular}
\end{table}


\paragraph{Anisotropic model.} Testing Eq.~(\ref{eq:W1M_anisotropic}) on the directional asymmetric demand under the Manhattan metric yields a better fit ($R^2_{\rm log}=0.99$, MAPE 6.78\%). The calibrated coefficient $\tilde C_M^W=0.2337$ is smaller than $\tilde C_M^L=0.5047$ but remains material.

\section{Case Study: New York City}
\label{sec:case_study}

This section applies the Manhattan Model (Eq.~\ref{eq:W1M}) to a
real-world mobility system: the New York City (NYC) for-hire vehicle
(FHV) network over January 2026. Comparable or better fits are obtained for other months (e.g., September 2025); January is selected because its outlier events provide a more demanding test of the model. 

Two nested scenarios are studied: (i)~the {full
NYC service area} (all five boroughs) and (ii)~{Manhattan Island only}
(a geographically constrained, high-demand sub-zone).
Both scenarios share the same theoretical model but differ in service
area, zone count, demand volume, aspect ratio, and --- crucially ---
the calibrated constant $C_{\rm M}$, enabling a controlled inter-scenario
comparison. 
The objective of the case study is to assess how the Manhattan Model predicts daily $\W$ values on a real
street network computed via exact LP. {The Euclidean Model is not used 
because the LP cost matrix is built from real-world street network distances, making the Manhattan metric the natural comparator.}


\subsection{Data and Study Design}
\label{sec:data}

Trip records are drawn from the NYC Taxi and Limousine Commission (TLC)
public trip data for January 2026 (31 days) \citep{nyc_tlc_2026}.
Each record contains a pickup zone, a drop-off zone, and a timestamp.
Daily origin and destination counts are aggregated to the 263 official
TLC taxi zones, which serve as the spatial units for the LP and model.
Table~\ref{tab:zones} summarises zone counts, trip shares, and key
activity hubs by borough. 
We note that NYC for-hire vehicles are dispatched on demand rather than
rebalanced as a shared fleet; the TLC trip data are used here to demonstrate that the
proposed $\W$ formula generalizes to real-world street networks and
empirical demand distributions.

\begin{table}[!ht]
\centering
\caption{NYC Taxi Zone breakdown by borough used in this study.}
\label{tab:zones}
\begin{tabular}{llll}
\toprule
\textbf{Borough} & \textbf{Zones} & \textbf{Trip share} &
\textbf{Key hubs} \\
\midrule
Manhattan    & 69 & 56\%  &
  Penn Station, Grand Central, Times Sq.\\
Queens       & 69 & 22\%  & JFK Airport, LaGuardia Airport\\
Brooklyn     & 61 & 15\%  & ---\\
Bronx        & 43 &  6\%  & ---\\
Staten Island& 21 &  1\%  & ---\\
\midrule
\textbf{Total} & \textbf{263} & \textbf{100\%} & \\
\bottomrule
\end{tabular}
\end{table}


We consider two scenarios defined by the service area:
\begin{itemize}[leftmargin=2em,itemsep=2pt]
  \item {Scenario (i) --- NYC (all boroughs):} all 263 zones
        are included.
        The bounding box is approximately square ($\kappa=1.056$),
        area $R=1830.2$\,km$^2$, giving $g_{\rm M}=2.001$.
  \item {Scenario (ii) --- Manhattan Island:} only trips with
        both pickup and drop-off within the 69 Manhattan zones are
        retained.
        The island's elongated north--south geometry yields $\kappa=2.393$
        and $g_{\rm M}=2.193$. 
\end{itemize}

Table~\ref{tab:geometry} summarises the bounding-box geometry for
each scenario.

\begin{table}[H]\centering
\caption{Bounding-box geometry for each scenario.}\label{tab:geometry}
\begin{tabular}{lllllll}\toprule
\textbf{Scenario} & \textbf{Zones} & $L$ (km) & $W$ (km) & $R$ (km$^2$) & $\kappa$ & $\gM$\\\midrule
NYC (all boroughs) & 263 & 43.965 & 41.628 & 1830.2 & 1.0562 & 2.0007\\
Manhattan Island   &  69 & 21.741 &  9.086 &  197.5 & 2.3927 & 2.1933\\
\bottomrule\end{tabular}\end{table}

For each day $t$ and scenario, the Wasserstein-1 distance
$\W^{(t)}$ is obtained by solving the transportation
LP~\eqref{eq:transport_LP}, where $d_{ij}$ is the shortest-path
distance (in km) between zone centroids $i$ and $j$ on the
OpenStreetMap (OSM) road network (computed via Dijkstra's algorithm),
and $p_i$, $q_j$ are the normalized pickup and drop-off fractions for
each zone.

The 31 days are divided into a 20-day \emph{calibration} set (the first
20 days of January) and an 11-day \emph{out-of-sample validation} set
(the remaining days).
The calibration set is used exclusively to estimate $\widehat{C}_{\rm M}$;
the validation set tests out-of-sample predictive accuracy.


Figures~\ref{fig:map_nyc} and~\ref{fig:map_man} visualise the spatial
structure of the rebalancing problem on a representative day
(Day~1, Thursday 1 January 2026).
Zone centroids are colored by net imbalance
$\rho_i $ (red: surplus, blue: deficit), and the 20
highest-flow LP-optimal rebalancing corridors are shown as purple
arrows.

\begin{figure}[H]\centering
\includegraphics[width=\linewidth]{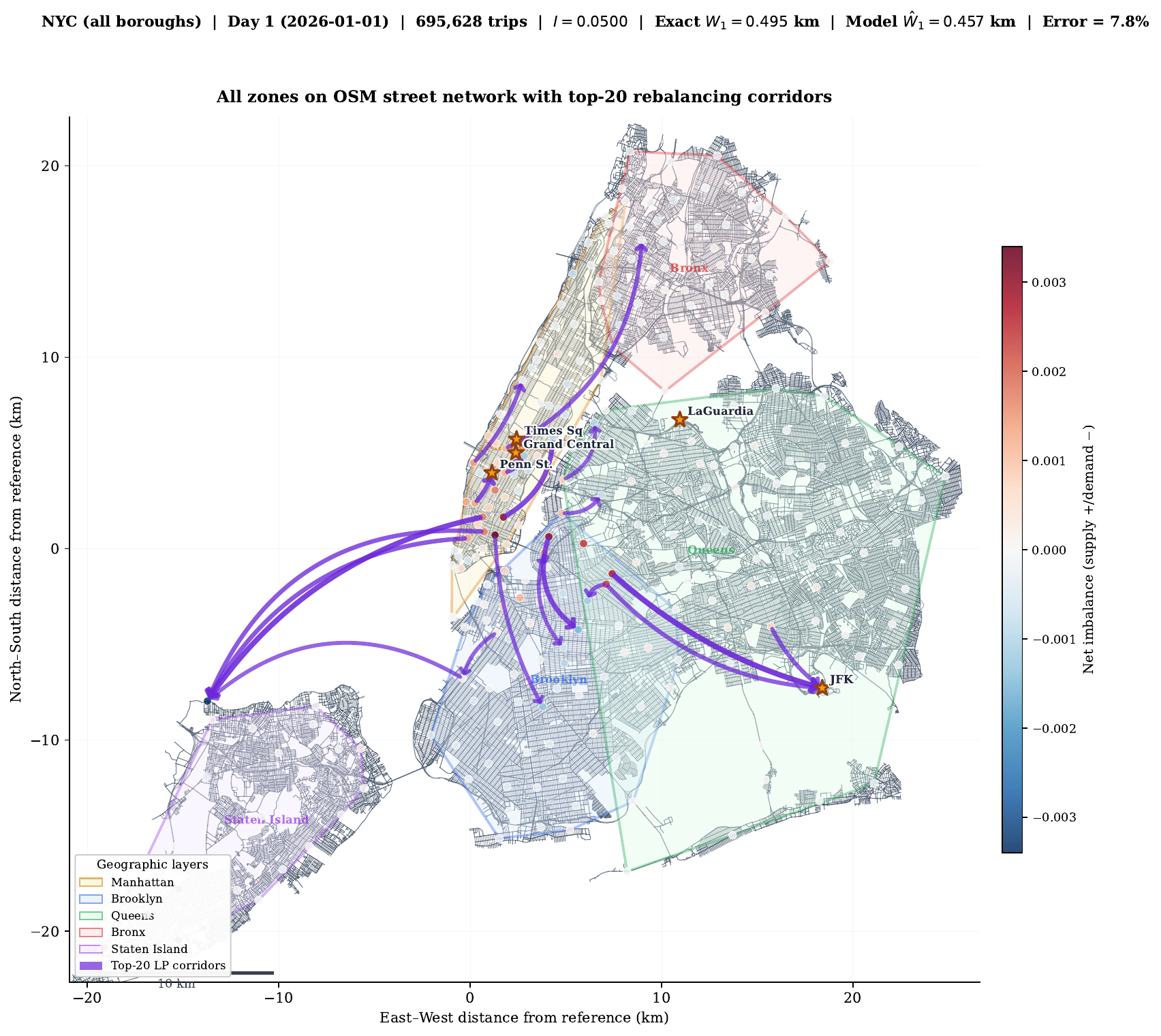}
\caption{\textbf{NYC All Boroughs --- Day-1 spatial map.}
OSM road network (gray edges), zone centroids colored by net imbalance
$\rho_i $ (\textcolor{red}{red}~=~surplus,
\textcolor{trblue}{blue}~=~deficit), purple arrows~=~top-20 LP-optimal rebalancing
corridors scaled by flow volume.}
\label{fig:map_nyc}
\end{figure}


\begin{figure}[H]\centering
\includegraphics[width=\linewidth]{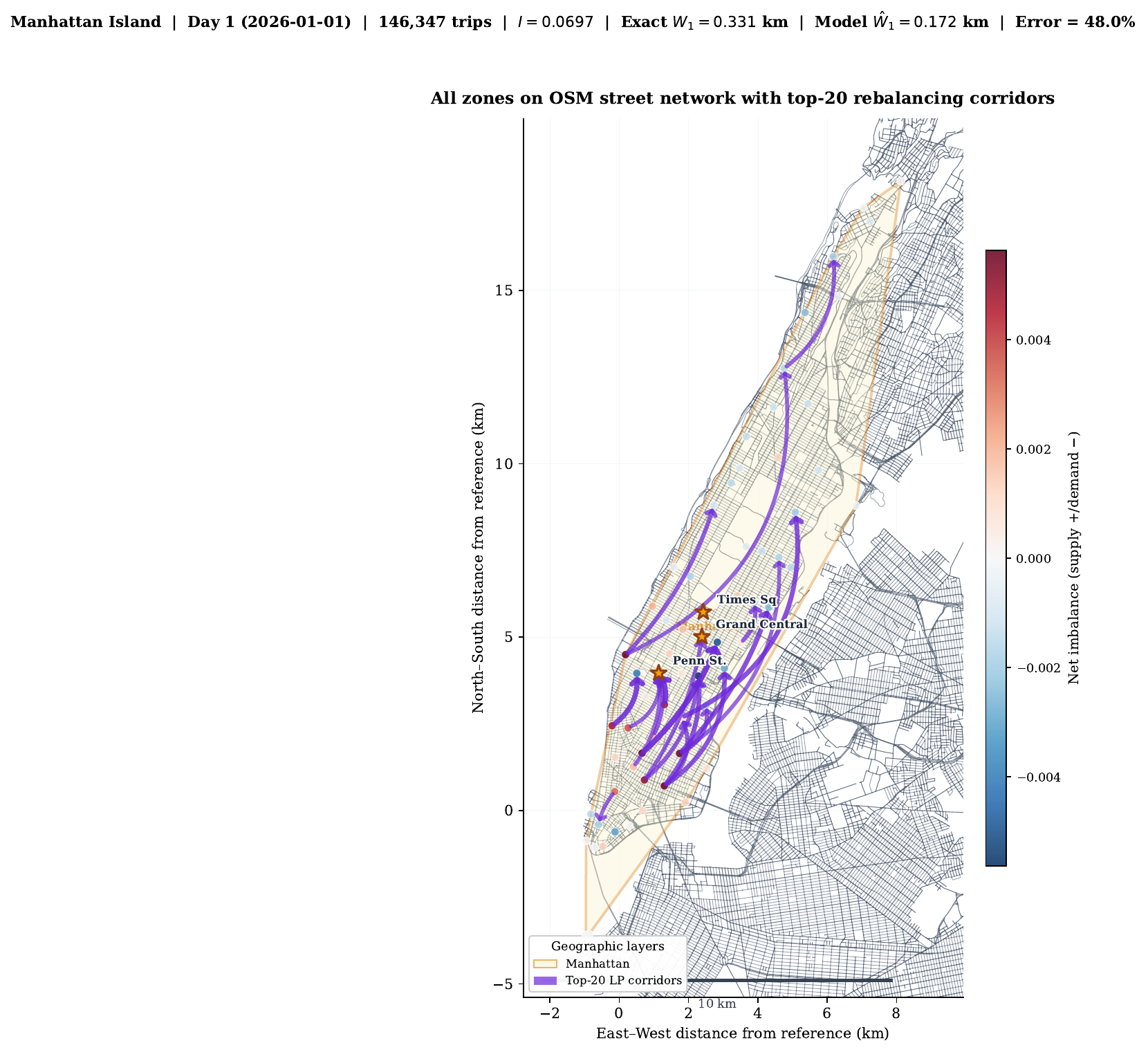}
\caption{\textbf{Manhattan Island --- Day-1 spatial map.}
Only trips with both origin and destination within Manhattan are included.
Zone imbalance and top-20 LP corridors shown on the OSM street network.}
\label{fig:map_man}
\end{figure}


\subsection{Descriptive Statistics}
\label{sec:desc}

Table~\ref{tab:desc} reports descriptive statistics for both scenarios
across the full 31-day study period.
Several features of Table~\ref{tab:desc} deserve comment.
Both scenarios exhibit low absolute imbalance ($I\approx 0.04$).
This reflects the high spatial density of NYC taxi trips (648k trips/day
across 263 zones), which produces substantial local cancellation between
nearby pickups and drop-offs.

The mean $\W^{\rm OSM}$ is 0.372\,km for NYC and 0.114\,km for
Manhattan.
These values appear small, but they are per-unit-probability-mass
distances between normalized distributions, not per-trip distances.
The physically interpretable quantity is $\W/I$, which equals the
average distance a vehicle must travel \emph{per unit of imbalance}:
$0.372/0.038 \approx 9.8$\,km for NYC and $0.114/0.040 \approx
2.9$\,km for Manhattan.
The factor-of-3.4 difference is consistent with the factor-of-3.0
difference in $\sqrt{R}$ between the two scenarios
($\sqrt{1830.2}/\sqrt{197.5}\approx 3.0$), confirming the $\sqrt{R}$
area scaling predicted by the model.


The per-day calibration ratio $\hat{C}_{\rm M}^{(t)} =
\W^{(t)}/(I^{(t)}\sqrt{R}\,g_{\rm M})$ has a mean of 0.1134
for NYC and 0.0839 for Manhattan.
Both values are below the $\hat{C}_{\rm M}=0.1440$ calibrated in the
synthetic numerical study (Section~\ref{sec:numerical}), which used
uniformly random demand on a rectangle.
Real NYC demand is spatially clustered and follows predictable
origin--destination patterns (commute corridors, airport transfers),
creating more local cancellations than uniform random demand and thus
smaller effective $\hat{C}_{\rm M}$.
The lower $\hat{C}_{\rm M}$ for Manhattan (0.0839 vs.\ 0.1134) further
suggests that Manhattan's highly structured demand---concentrated
along transit corridors with predictable inbound/outbound
flows---produces even more efficient local rebalancing cancellations
than the multi-borough NYC pattern.

\begin{table}[H]\centering
\caption{Descriptive statistics over all 31 days, by scenario.}\label{tab:desc}
\begin{tabular}{llrrrr}\toprule
\textbf{Scenario} & \textbf{Metric} & \textbf{Mean} & \textbf{Std} & \textbf{Min} & \textbf{Max}\\\midrule
\multirow{4}{*}{\textbf{NYC}} & Trips $N$ & 648,116 & 118,463 & 310,746 & 851,637\\
 & Imbalance $I$ & 0.0381 & 0.0088 & 0.0278 & 0.0721\\
 & Exact $\W^{\rm OSM}$ (km) & 0.372 & 0.108 & 0.204 & 0.611\\
 & Per-day $C_{\rm M}$ & 0.1134 & 0.0181 & 0.0824 & 0.1553\\
\midrule
\multirow{4}{*}{\textbf{Manhattan}} & Trips $N$ & 164,469 & 38,668 & 76,649 & 250,256\\
 & Imbalance $I$ & 0.0404 & 0.0136 & 0.0288 & 0.0947\\
 & Exact $\W^{\rm OSM}$ (km) & 0.114 & 0.089 & 0.049 & 0.473\\
 & Per-day $C_{\rm M}$ & 0.0839 & 0.0277 & 0.0545 & 0.1620\\
\bottomrule
\end{tabular}\end{table}

\noindent

\subsection{Calibration and Validation}
\label{sec:calib_cm}

The calibration constant is estimated from the 20-day calibration set
using robust median as in the previous section. 
Figure~\ref{fig:ts_both} presents the daily time series of exact
$\W^{\rm OSM}$ (solid lines) and model predictions (dashed lines)
for both scenarios across all 31 days.

\begin{figure}[H]\centering
\includegraphics[width=\linewidth]{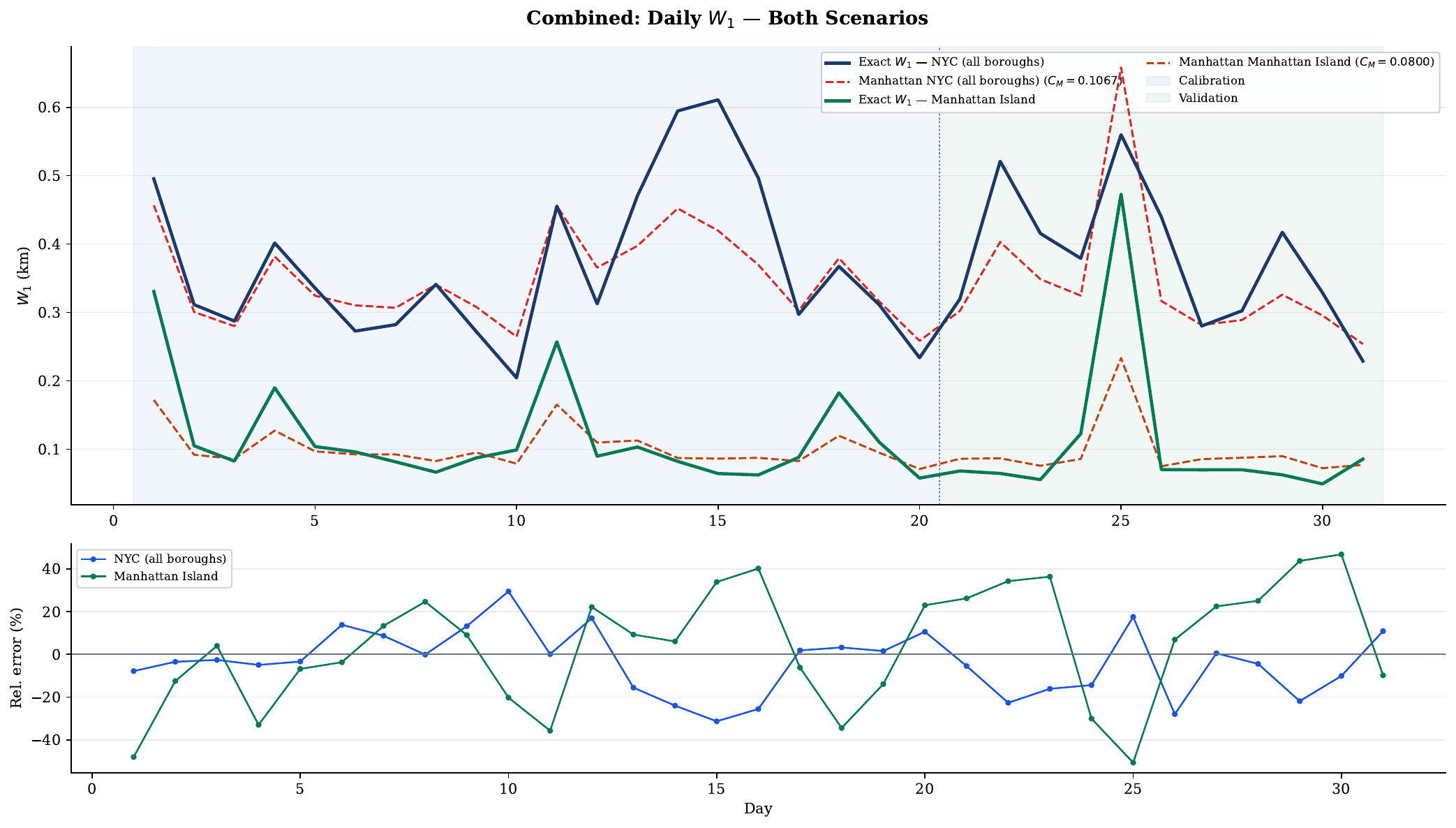}
\caption{\textbf{Daily $\W$ time series for both scenarios.}
Solid markers: exact LP solution on OSM network.
Dashed line: Manhattan Model prediction using calibrated $C_{\rm M}$.
Vertical dashed line separates the calibration (left) and validation (right) periods.
Note the pronounced scale difference: NYC $\W$ is approximately 3.3$\times$ larger
than Manhattan, consistent with the $\sqrt{R}$ area scaling.}
\label{fig:ts_both}
\end{figure}



Both scenarios exhibit recurring weekly cycles, though with opposite
weekend effects.
Two outlier days stand out: New Year's Day (1 January) elevates $\W$
in both scenarios yet the model under-predicts by 7.8\% (NYC) and
48.0\% (Manhattan) because holiday demand produces a spatially
concentrated imbalance not captured by the model's linear
$I$-dependence; and 25 January---a Sunday with a major snowstorm \citep{nws_jan2026_storm}
suppressing Manhattan trips to 54\% below the weekend average---yields
the largest single-day relative error (50.6\%).
These extreme days are identifiable ex post from operational records and
indicate that $\widehat{C}_{\rm M}$ should be re-estimated when demand
conditions deviate substantially from historical norms.

Table~\ref{tab:errors_cs} and Figures~\ref{fig:parity_both} present the full prediction error suite for both
scenarios across both sets. The NYC scenario achieves a calibration MAPE of 10.9\% and a
validation MAPE of 13.8\%, a remarkably small degradation for an
out-of-sample period.
The Manhattan scenario shows higher errors (calibration MAPE 20.0\%,
validation MAPE 30.2\%).
These error levels reflect an inherent limitation of a single-scalar
approximation: the imbalance index~$I$ captures the \emph{magnitude}
of supply--demand mismatch but not its spatial arrangement, so
day-to-day fluctuations in spatial structure contribute to prediction
error independently of~$I$.
The model should therefore be interpreted as an estimator of the
\emph{expected} rebalancing cost for a calibrated demand regime---a
planning and benchmarking tool---rather than a day-to-day forecasting
instrument.
The log-scale $R^2$ values of 0.72 (NYC calibration) and 0.65
(Manhattan calibration) indicate that the model explains approximately
65--72\% of the variance in $\log\W$ using only $I$ as the day-to-day
driver (since $\sqrt{R}$ and $g_{\rm M}$ are constant within a
scenario).
The remaining unexplained variance reflects the spatial structure of
daily demand.
The MBE is positive in both scenarios and both sets: the model
systematically under-predicts $\W$. This bias likely arises because $C_{\rm M}$ was calibrated on synthetic instances
with spatially uncorrelated demand, whereas real NYC trips exhibit spatial correlation.
The scalar index $I$ captures total imbalance magnitude but not its spatial
arrangement.
The models therefore provide a lower bound on operational rebalancing cost
when demand is spatially structured. 

\begin{table}[H]\centering
\caption{Full error suite: Manhattan Model vs.\ exact OSM $\W$.}\label{tab:errors_cs}
\begin{tabular}{lrrrr}\toprule
 & \multicolumn{2}{c}{\textbf{NYC (i)}} & \multicolumn{2}{c}{\textbf{Manhattan (ii)}}\\
\cmidrule(lr){2-3}\cmidrule(lr){4-5}
\textbf{Metric} & \textbf{Cal.} & \textbf{Val.} & \textbf{Cal.} & \textbf{Val.}\\\midrule
MAE (km)            & 0.044 & 0.058 & 0.029 & 0.039 \\
RMSE (km)           & 0.068 & 0.072 & 0.047 & 0.075 \\
{MAPE (\%)}  & ${10.91}$ & ${13.81}$
                    & ${20.00}$ & ${30.21}$ \\
$R^2$(log)          & 0.721 & 0.531 & 0.654 & 0.679 \\
MBE (km)            & 0.018 & 0.036 & 0.015 & 0.012 \\
\bottomrule\end{tabular}\end{table}

\noindent
The parity plots in Figure~\ref{fig:parity_both} display the
point-by-point agreement between model and LP for each day.
For NYC (top row), points cluster tightly around the 1:1 line across
the full range of observed $\W$ values, with most days falling within
the $\pm$10\% during calibration and the majority
remaining within $\pm$20\% in validation.
For Manhattan (bottom row), the spread is wider and a few outlier days
are visible --- particularly in the low-$\W$ (high-volume weekday)
and high-$\W$ (low-volume weekend and holiday) extremes.
The pattern of larger relative errors at extreme $\W$ values is
consistent with the model's use of a constant $C_{\rm M}$: the
calibration median optimally fits the bulk of the distribution but
underfits on anomalous days where the demand structure differs
qualitatively from the average.
\begin{figure}[H]\centering
\includegraphics[width=\linewidth]{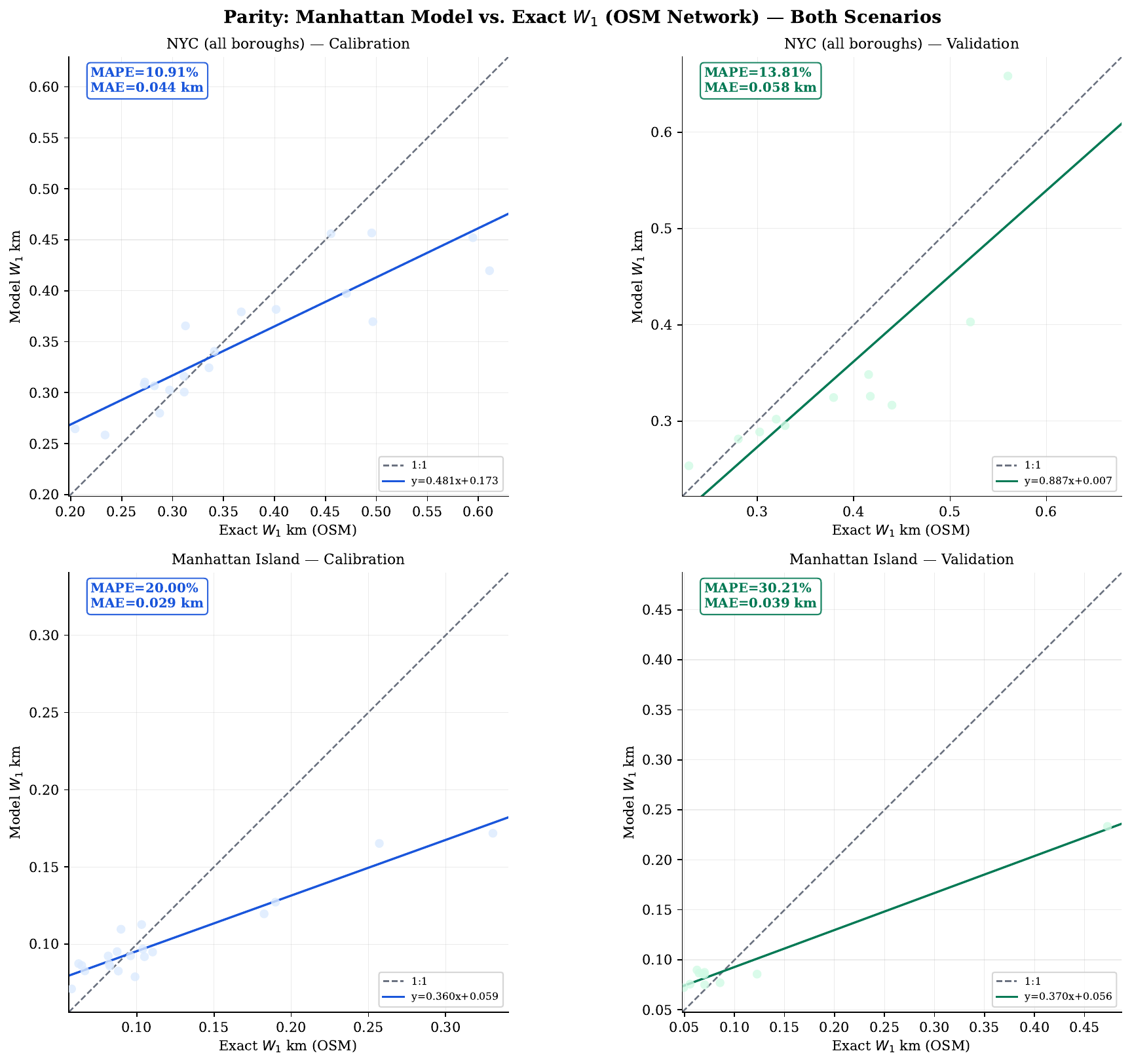}
\caption{\textbf{Parity plots: Manhattan Model vs.\ exact OSM $\W$,
  both scenarios.}
  $2\times 2$ panel grid: rows = scenarios (NYC top, Manhattan bottom),
  columns = sets (calibration left, validation right).
  Each point is one day; dashed line = 1:1 reference. Points above the 1:1 line indicate under-prediction.}
\label{fig:parity_both}
\end{figure}



\subsection{Rebalancing Vehicle-Kilometers}
\label{sec:vkm}

The per-vehicle distance $\W$ is a normalized metric; the
operationally relevant quantity for fleet management is the
total rebalancing VKT, defined as:
\begin{equation}
  \mathrm{VKT}_t \;=\; \W^{(t)} \times N^{(t)},
  \label{eq:vkm}
\end{equation}
where $N^{(t)}$ is the number of trips on day $t$.
This quantity directly translates the model into a cost-planning metric:
it represents the minimum total kilometers that empty vehicles must
travel to restore fleet balance, assuming a perfect routing algorithm.




Table~\ref{tab:vkm_summary} summarises VKT statistics for both scenarios.
The NYC system requires a mean of approximately {238,000\,km of empty
vehicle travel per day} to achieve optimal rebalancing.
Over the full 31-day study period, the total minimum rebalancing burden
reaches {7.38 million km}.
Manhattan in isolation accounts for a mean of 17,350\,km/day
(537,856\,km total), which is 7.3\% of the NYC total despite generating
56\% of all trips.
This disproportionately small VKT share reflects Manhattan's lower
effective $C_{\rm M}$ and smaller service area: the highly dense, short-range
trip structure of Manhattan demand leads to efficient local cancellations
that dramatically reduce the net rebalancing requirement.
Note again that these figures represent the \emph{theoretical minimum} empty travel
implied by the observed demand imbalance, not actual rebalancing
operations (FHV fleets are not centrally rebalanced).
\begin{table}[H]\centering
\caption{Rebalancing VKT statistics by scenario (31 days, January 2026).}
\label{tab:vkm_summary}
\renewcommand{\arraystretch}{1.3}
\begin{tabular}{lll}\toprule
\textbf{Statistic} & \textbf{NYC (all boroughs)} & \textbf{Manhattan Island}\\\midrule
Mean VKT/day (km)    & 238,014 & 17,350 \\
Std dev (km)         &  75,498 & 10,259 \\
Min VKT/day (km)     & 147,825 &  6,369 \\
Max VKT/day (km)     & 410,829 & 48,386 \\
Total over 31 days (km) & 7,378,429 & 537,856 \\
Weekday mean (km)    & 243,115 & 13,595 \\
Weekend mean (km)    & 225,544 & 26,529 \\
VKT ratio (NYC/Manhattan) & \multicolumn{2}{c}{$13.7\times$} \\
\bottomrule\end{tabular}\end{table}

\section{Conclusion}
\label{sec:conclusion}

We derived closed-form approximation models for the minimum rebalancing distance of shared mobility systems in rectangular regions, for both the Manhattan (grid road) and Euclidean (direct) metrics. 
The key novel contributions include: (i) the incorporation of an imbalance index that captures the structured mismatch between supply and demand; (ii) the unified model with a shared shape factor for both Manhattan and Euclidean metrics; and (iii) the validation through extensive numerical studies and a real-world case study against exact LP solutions, demonstrating their practical accuracy and operational relevance. 
These results complement the traditional LP by giving shared mobility operators a simple, theoretically grounded benchmarking formula to guide their medium to long-term operational decisions, while the traditional LP can be used for short-term and case-specific optimization generating exact transport plans. 


Several policy implications can be drawn from our models. First, the $\sqrt{R}$ scaling quantifies the benefit of a common industry practice: partitioning the service area into sub-regions and confining vehicles within each. The model implies that subregion shape affects the minimum rebalancing cost and should be chosen to reflect the demand distribution; for isotropic demand, square regions minimize that cost. Second, the operators can use the model to determine the optimal rebalancing frequency $h$ by balancing the cost of vehicle repositioning against the benefit of reduced imbalance. Third, incentive programs that reduce imbalance $I$ can have a predictable, proportional impact on rebalancing VKT, which can be quantified using the model. Fourth, operators should calibrate specific $C$ for typical demand patterns that differ significantly from each other. {They may also adopt model variants, e.g., the anisotropic model as given Appendix \ref{appx_anisotropic}.} Finally, these models provide a computationally inexpensive benchmark for optimizers; if an operator's extension to new regions cannot profit against this benchmark, it may indicate an non financially viable expansion.

More broadly, the value of these closed-form approximation models extends well beyond
shared-mobility rebalancing. Any systems that must correct a spatial or
network imbalance using limited relocation effort can benefit from the same
approximate analytical logic, including inventory repositioning, humanitarian relief
distribution, warehouse robot coordination, and service-capacity balancing
problems. The proposed models provide a compact
way to link demand heterogeneity, geometry, and relocation cost, which can
help other fields obtain interpretable scaling laws and fast benchmarks without repeatedly solving large optimization models.

Lastly, we acknowledge several limitations and suggest directions for future research. The closed-form approximation, resting on Assumptions A1--A3, summarizes spatial imbalance through the scalar index $I$, which captures the total magnitude of surplus and deficit but not their spatial arrangement. Our numerical results show robust performance across demand distributions with different spatial arrangements, achieved through the calibrated constant $C$ that absorbs this effect. Nevertheless, $C$ must be calibrated empirically for stable demand distributions rather than predicted analytically. Future work could incorporate a spatial correlation term into $I$ to capture the effect of distance between surplus and deficit locations on rebalancing cost. Additionally, the current derivations assume a rectangular service region, which may not reflect real-world areas that are often irregularly shaped. Extending the theory to convex and concave polygons would broaden the models' applicability.
The model also assumes rebalancing cost proportional to transport mass and distance, which excludes economies of scale arising from large-capacity vehicles used for bulk repositioning.
Extending the framework to subadditive cost structures — for example, a hierarchical scheme in which a truck layer handles long-range bulk flows while individual vehicles cover local last-mile repositioning — is a natural direction for future work.

\section*{Acknowledgments}

The authors used AI-assisted tools, including Claude (Anthropic) and GitHub Copilot, to support drafting and editorial refinement during the preparation of this manuscript.

\appendix
\section{An Anisotropic Extension}
\label{appx_anisotropic}
To preserve computational simplicity, we define directional imbalance proxies from the same zonal demand table:
For a grid with columns $c=1,\dots,n_x$ and rows $r=1,\dots,n_y$, let
\begin{align}
  p_c^x = \sum_r p_{cr}, \qquad q_c^x = \sum_r q_{cr}
\end{align}
be the origin and destination shares aggregated by column, and
\begin{align}
  p_r^y = \sum_c p_{cr}, \qquad q_r^y = \sum_c q_{cr},
\end{align}
be the origin and destination shares aggregated by row. Then define
\begin{align}
  \tilde I_x = \frac{1}{2}\sum_{c=1}^{n_x} |q_c^x - p_c^x|,
  \qquad
  \tilde I_y = \frac{1}{2}\sum_{r=1}^{n_y} |q_r^y - p_r^y|.
\end{align}
The anisotropic extension of the Manhattan model is then:
\begin{equation} \label{eq:W1M_anisotropic}
  W_{1,\mathrm{anisotropic}}^{\mathrm{M}}
  \approx
   \tilde C_M^L \tilde I_x L + \tilde C_M^W \tilde I_y W,
\end{equation}
which will be evaluated against the exact LP solution in Section \ref{sec:numerical}.

\section{Derivation of Euclidean Model} \label{appx_proof}

    We derive the Euclidean model in three steps combining exact
    analysis with heuristic scaling arguments.
    Step~1 computes the transport cost for the two dominant Fourier modes
    of $\rho$ using Beckmann's primal flow formulation~\eqref{eq:W1E_Beckmann}.
    Steps~2--3 invoke scaling approximations analogous to
    Assumptions~A1--A3 in the Manhattan derivation to combine the mode
    costs and relate $C_{\rm E}$ to $C_{\rm M}$.

    \bigskip
\noindent\emph{Step 1 — Dominant Fourier mode analysis.}

\smallskip\noindent
As established by Beckmann's flow formulation~\eqref{eq:W1E_Beckmann},
the Euclidean rebalancing cost equals the minimum $L^2$-norm vector flux
satisfying $\nabla\cdot\mathbf{v}=\rho$ with zero normal flux at the boundary.
We approximate this cost by analyzing the two lowest non-trivial Fourier
cosine modes of $\rho$ on $\Omega=[0,L]\times[0,W]$, which form a
complete basis under Neumann boundary conditions.

\medskip
\noindent\textit{Mode $(1,0)$: pure east--west imbalance.}\;
Take $\rho(x,y) = A\cos(\pi x/L)$.
Normalizing via the imbalance index~\eqref{eq:I}
gives $A = \pi I/R$.
Since $\rho$ has no $y$-variation, all transport is one-dimensional
along $x$.
The marginal imbalance profile is
\[
  \rho_x(x) = \int_0^W \rho(x,y)\,\dd y
             = WA\cos\!\Bigl(\tfrac{\pi x}{L}\Bigr)
             = \frac{\pi I}{L}\cos\!\Bigl(\tfrac{\pi x}{L}\Bigr).
\]
The optimal 1-D flow satisfies $\dd v_x/\dd x = \rho_x$ with
$v_x(0)=v_x(L)=0$, giving:
\[
  v_x^*(x)
  \;=\; \int_0^x \rho_x(s)\,\dd s
  \;=\; I\sin\!\Bigl(\tfrac{\pi x}{L}\Bigr).
\]
The transport cost for this mode is:
\begin{equation}
  {\W^{\rm E}}_{(1,0)}
  \;=\; \int_0^L\bigl|v_x^*(x)\bigr|\dd x
  \;=\; I\int_0^L \sin\!\Bigl(\tfrac{\pi x}{L}\Bigr)\dd x
  \;=\; \frac{2IL}{\pi}.
  \label{eq:cost10}
\end{equation}
The cost grows with $L$, i.e.\ with $\kappa$.

\medskip
\noindent\textit{Mode $(0,1)$: pure north--south imbalance.}\;
By symmetry, $\rho(x,y)=A\cos(\pi y/W)$ gives
$v_y^*(y) = I\sin(\pi y/W)$, and:
\begin{equation}
  {\W^{\rm E}}_{(0,1)}
  \;=\; \int_0^W\bigl|v_y^*(y)\bigr|\dd y
  \;=\; \frac{2IW}{\pi}.
  \label{eq:cost01}
\end{equation}
The cost decreases with $\kappa$ (since $W=\sqrt{R/\kappa}\to 0$).

\medskip
\noindent\textit{Combining both modes (heuristic approximation).}\;
When the imbalance $\rho = f - g$ has no preferred spatial orientation
--- as occurs when origins and destinations are not systematically
aligned along either axis --- both modes are excited with comparable
amplitude and neither dominates.
Note that $W_1$ is not a linear functional of $\rho$: the transport
cost of a sum of modes is not in general the sum of the individual
mode costs.  As a heuristic scaling approximation, we treat the two
axis-aligned cost contributions as additive, giving
$\W^{\rm E}\approx\frac{2I}{\pi}(L+W)$.
Absorbing the factor $2/\pi$ and residual higher-mode contributions
into the calibration constant $C_{\rm E}$:
\begin{equation}
  \W^{\rm E}
  \;\approx\; C_{\rm E}\cdot I\cdot(L+W).
  \label{eq:W1E_intermediate}
\end{equation}

\begin{remark}[Approximation scope]
This heuristic rests on two related approximations.
First, the additive combination of the two axis-aligned modes implicitly
assumes that both modes carry comparable amplitude --- i.e., that the
imbalance field has no strongly preferred spatial orientation
({Assumption~A4}: near-isotropy of demand orientation).
For strongly anisotropic demand (e.g., systematic directional flow
concentrated along one axis), $C_{\rm E}$ would need to be split into
axis-specific constants $C_{\rm E}^L$ and $C_{\rm E}^W$, analogous to
the anisotropic Manhattan extension in Appendix~\ref{appx_anisotropic}.
Second, cross-modes $(m,n)$ with $m,n\geq 1$ carry additional transport
cost along oblique directions that is not modelled explicitly; their
contribution is absorbed empirically into $C_{\rm E}$.
For demand patterns with strong diagonal structure (e.g., surplus and
deficit concentrated in diagonally opposite quadrants), the two-mode
approximation is less accurate, consistent with the higher MAPE of the
Euclidean model relative to Manhattan observed in
Section~\ref{sec:numerical}.
\end{remark}

\bigskip
\noindent\emph{Step 2 — The shape factor $g_{\rm M}(\kappa)$.}

\smallskip\noindent
Substituting $L=\sqrt{\kappa R}$ and $W=\sqrt{R/\kappa}$
into~\eqref{eq:W1E_intermediate}:
\begin{equation}
  L + W
  \;=\; \sqrt{R}\!\left(\sqrt{\kappa}+\frac{1}{\sqrt{\kappa}}\right)
  \;=\; \sqrt{R}\cdot\frac{\kappa+1}{\sqrt{\kappa}}
  \;=\; \sqrt{R}\cdot g_{\rm M}(\kappa),
  \label{eq:LplusW_kappa}
\end{equation}
where $g_{\rm M}(\kappa)=(\kappa+1)/\sqrt{\kappa}$ is the
{shape factor} --- the same function that appears in
the Manhattan model (Section~\ref{sec:FinalManhattanModel}).


\bigskip
\noindent\emph{Step 3 — Euclidean model.}

\smallskip\noindent
Substituting~\eqref{eq:LplusW_kappa} into~\eqref{eq:W1E_intermediate}
yields:
\[
  \W^{\rm E}
  \;\approx\; C_{\rm E}\cdot I\cdot\sqrt{R}\cdot g_{\rm M}(\kappa),
\]
which is Eq.~\eqref{eq:W1E}.

\medskip
\noindent\textit{Isotropic-direction heuristic for $C_{\rm E}/C_{\rm M}$.}\;
Consider a random vehicle relocation with Euclidean distance $r$ and
direction $\theta$ uniformly distributed on $[0,2\pi)$ (isotropic demand).
The Manhattan distance for the same trip is
$d_{\rm M} = r\,(|\cos\theta|+|\sin\theta|)$.
Since
\[
  \mathbb{E}_\theta\bigl[|\cos\theta|+|\sin\theta|\bigr] = \frac{4}{\pi},
\]
the expected ratio is
\[
  \frac{\mathbb{E}[d_{\rm E}]}{\mathbb{E}[d_{\rm M}]}
  \;=\; \frac{\pi}{4}
\]
exactly under isotropic demand.
This is an exact result for the single-trip direction distribution; its
extension to the aggregate rebalancing cost is a \emph{heuristic} that
assumes the distribution of optimal-transport trip directions is
approximately isotropic (Assumption~A4 above).
For uniform random points in a rectangle, isotropy holds approximately
(boundary effects introduce slight anisotropy); for instance, direct
computation for the unit square gives the ratio $\approx 0.7821$,
within $0.4\%$ of $\pi/4\approx 0.7854$.
Since both models share the same shape factor $g_{\rm M}(\kappa)$,
this ratio applies approximately to the calibration constants:
\begin{equation}
  C_{\rm E} \;\approx\; \frac{\pi}{4}\,C_{\rm M}
              \;\approx\; 0.785\times 0.1440
              \;\approx\; 0.113.
  \label{eq:CE_CM}
\end{equation}
This prediction is a heuristic; its validity is confirmed numerically:
the calibrated value $\hat{C}_{\rm E}=0.1189$
(Section~\ref{sec:numerical}) is within 6\% of the prediction,
consistent with the near-isotropic demand patterns in the numerical
study and validating Assumption~A4 in this setting.

\bigskip

\bigskip
\bibliographystyle{abbrvnat}

\end{document}